\documentclass[prd,twocolumn,superscriptaddress,floatfix,amsmath,amssymb,amsfonts,nofootinbib,longbibliography]{revtex4-2}

\usepackage{float} 
\usepackage{scalerel}
\usepackage[normalem]{ulem}
\usepackage[english]{babel}
\usepackage{graphicx}
\usepackage{dcolumn}% Align table columns on decimal point
\usepackage{bm}% bold math
\usepackage{blindtext}
\usepackage{verbatim}
\usepackage{relsize}
\usepackage{mathrsfs}
\usepackage{musicography}
\usepackage{amsmath}
\usepackage{blindtext}
\usepackage{cancel}
\usepackage{physics}
\usepackage{epstopdf}
\usepackage{mathtools}
\usepackage{blindtext}
\usepackage{tensor}
\usepackage{color}
\usepackage[usenames,dvipsnames]{pstricks}
\usepackage{epsfig}
\usepackage{pst-grad}
\usepackage{pst-plot} 
\usepackage{hyperref}
\usepackage{verbatim}
\usepackage{slashed}
\usepackage{dsfont}
\usepackage[english]{babel}
\usepackage{amsmath,amssymb,amsthm} 
\usepackage{lipsum}

\NewDocumentEnvironment{alignb}{b}{%
  \begin{align*}
 ~\refstepcounter{equation} #1 \tag{\theequation}
  \end{align*}
}{\ignorespacesafterend}

\newcommand{\w}{\omega}
\newcommand{\xx}{\mf{x}}

\newcommand{\mf}{\mathsf}%\mathsf is too long

\newcommand{\ii}{\mathrm{i}}
\newcommand{\ee}{\mathrm{e}}

\newcommand{\tc}[1]{\textsc{#1}}

\newcommand{\trans}[1]{{#1}^{\!\mathsf{T}}}
\newcommand{\ptrans}[1]{{#1}^{\!\mathsf{T}_\tc{B}}}

\newcommand{\sct}[1]{\scaleto{#1}{3.2pt}}

\allowdisplaybreaks[1] %<<<<<<<<<<<<<<<<<<<<<<
\newtheorem{proposition}{Proposition}

\begin{document}

\title{Optimization of entanglement harvesting with arbitrary temporal profiles:\\
the limit of second order perturbation theory}

\author{Marcos Morote-Balboa}
\email{mamo2246@student.su.se}

\affiliation{Nordita,
KTH Royal Institute of Technology and Stockholm University,
Hannes Alfvéns väg 12, 23, SE-106 91 Stockholm, Sweden}

\affiliation{Department of Physics, Stockholm University, SE-106 91 Stockholm, Sweden}

\author{T. Rick Perche}
\email{rick.perche@su.se}

\affiliation{Nordita,
KTH Royal Institute of Technology and Stockholm University,
Hannes Alfvéns väg 12, 23, SE-106 91 Stockholm, Sweden}

\begin{abstract}
    We study the protocol of entanglement harvesting when two local probes couple to the vacuum of a real scalar quantum field with arbitrary temporal profiles. We use a Hermite expansion to efficiently compute smeared field propagators in closed-form, recasting the negativity between the probes as a matrix product. We then optimize the protocol under different signalling conditions, enhancing entanglement harvesting by several orders of magnitude. This optimization would take current experimental proposals beyond the regime of second order perturbation theory.

   %When applying our results to different setups where the detectors are at different levels  of spacelike separation, our optimization enhances the protocol by several orders of magnitude, while keeping a negligible, or even non-existent, signalling.
   
    % By developing a Hermite-expansion method for quick computation of the final entanglement between the probes

    % in when the detectors are at different levels of spacelike separation.
   
   %When the probes are spacelike separated we find that oscillatory temporal profiles optimize the protocol, exemplifying the fact that superoscillating functions can push entanglement harvesting to the limit of perturbation theory. When allowing non-negligible but small signalling between the probes, we find a three order of magnitude increase in the negativity compared to a Gaussian switching function. Finally, we find parameters that prevent communication from causally connected probes, allowing them to harvest entanglement to the limit of second order perturbation theory without the need of superoscillations.
\end{abstract}

\maketitle

\section{Introduction}\label{section:intro}

Although the vacuum of a quantum field theory (QFT) shares an infinite amount of entanglement between a region and its complement~\cite{embezzlers2024}, there is, to date, no experimental evidence of spacelike vacuum entanglement. This contrast is due to the fact that, to access entanglement between complementary regions, one would require probing arbitrarily high energy modes localized at the boundary~\cite{areaLaw1993,areaLawReview2010,witten,kelly}. Physically, we are then limited to probing entanglement between bounded non-complementary regions, which is finite and decays exponentially with their separation~\cite{finiteEnt2009,KlcoUVIR,KlcoEntAllDist}. The most promising attempt in the direction of witnessing vacuum entanglement is the protocol of entanglement harvesting~\cite{Valentini1991,Reznik1,Pozas-Kerstjens:2015}, where local probes interact with a quantum field, extracting entanglement from the degrees of freedom that they couple to.

%spacetime regions where the interaction takes place.

% where two local probes couple to independent degrees of freedom and extract pre-existing entanglement from the field.

% couple to independent field degrees of freedom and extract entanglement a quantum field. 

In recent years, different proposals for experimental implementations of entanglement harvesting have been put forward, seeking to extract entanglement both from analogous systems~\cite{cisco2023harvesting,oberthaler} and from the electromagnetic field~\cite{Mo_Settembrini2022VacuumCorrelationsNatCommun,adamExperimentalEH2025}. Despite these proposals allowing probes to, in principle, extract entanglement from the vacuum, the resulting entanglement is typically inappreciable and barely large enough to be observed. While the proposals~\cite{cisco2023harvesting,oberthaler,adamExperimentalEH2025,Mo_Settembrini2022VacuumCorrelationsNatCommun} all consider setups that maximize entanglement extraction, these optimizations are typically limited to a few real parameters that cannot control the most essential part of local interactions in QFT: the profile of the interaction in spacetime. %In this manuscript we will show that considering a more general optimization of this sort can  

Indeed, despite the many entanglement harvesting setups considered in the literature (see e.g.~\cite{Valentini1991,Reznik1,reznik2,RalphOlson2,NickEdu2014,Pozas-Kerstjens:2015,Pozas2016,HarvestingQueNemLouko,HarvestingBHLaura,Cong2019,EricksonZero,HarvestingSuperposed,ericksonNew,HarvestingDelocalized,carol,HarvestingAccelerationRobb,hectorMass,boris,FullyRelativisticEH,adrianQuenchedHarv2025}), 
% Indeed, despite the many different theoretical setups in which entanglement harvesting has been considered (see e.g.~\cite{NickEdu2014,Pozas-Kerstjens:2015,Pozas2016,EricksonZero,HarvestingSuperposed,Cong2019,ericksonNew,carol,boris,waaaymore}), 
optimizations of the protocol usually consist of varying only a few parameters of the probes, such as their energy gap and separation in spacetime~\cite{RalphOlson2,Pozas-Kerstjens:2015,hectorMass,ericksonNew}, while keeping a constant temporal profile of the interaction (a single positive pulse,  typically a Gaussian~\cite{Pozas-Kerstjens:2015,Pozas2016,EricksonZero,HarvestingSuperposed,Cong2019,ericksonNew,carol,HarvestingAccelerationRobb,HarvestingBHLaura,hectorMass,boris,FullyRelativisticEH}). However, in~\cite{Reznik1}, it was argued that more general oscillatory switching functions could enhance the entanglement extracted. A complete optimization of entanglement harvesting must therefore also optimize over the spacetime profile of the interaction.

% allow one to 

% require one to treat even the most fundamental aspects of these detectors as degrees of freedom that can be optimized over.

%In this manuscript we will show that considering a more general optimization of this sort can 

%Despite the \trr{broad literature} on entanglement harvesting, the protocol has never been optimized in full generality. Existing studies typically optimize a few parameters within \trr{a fixed family of setups}, such as the detectors' energy gap, separation, or \trr{the width} of a prescribed switching function. However, the temporal profile of the interaction is itself a central ingredient of the protocol, \trr{and the spacetime regions where the detectors couple to the field are usually taken as fixed from the start}, \too{most commonly a single positive pulse such as a Gaussian}. 
%\tmm{A genuine optimization should therefore treat the switching function as part of the problem, rather than as an external choice.}

In this manuscript we will optimize the protocol of entanglement harvesting when considering arbitrary temporal profiles for the detectors-field coupling. We will carry this optimization over subspaces spanned by Hermite functions, which {will allow us to express the entanglement and signaling quantifiers between the detectors as simple matrix products, whose components we obtain in closed-form. When applying our results to different setups with detectors at different levels of spacelike separation, our optimization enhances the protocol by several orders of magnitude, while keeping a negligible, or even non-existent, signalling. To quantify the causal contact between the probes, we define the signalling-to-entanglement ratio, a stricter signalling quantifier than the previously proposed communication-mediated entanglement estimator~\cite{ericksonNew}. Extrapolating our results to experimental setups, we then argue that such an optimization would take entanglement harvesting beyond the regime of second order perturbation theory.

This manuscript is organized as follows. In Section~\ref{section:basics-aqft} we review propagators in a real scalar QFT and how they encode quantum correlations and communication. In Section~\ref{sec:iii} we review the protocol of entanglement harvesting and present an explicit example with Gaussian switching functions that will be used as a comparison throughout the manuscript. Section~\ref{sec:iv} presents our Hermite-expansion method for computing the relevant smeared propagators in QFT for arbitrary temporal profiles. In Section~\ref{sec:optimization} we apply our method to optimize the protocol of entanglement harvesting under different signalling conditions. Section~\ref{sec:limit} discusses the limitations of the second order perturbative approach to the protocol, as well as potential experimental consequences of our results. The conclusions of our work can be found in Section~\ref{sec:conclusions}.

\phantom{a}

\section{Propagators in Quantum Field Theory}\label{section:basics-aqft}

%There are many ways of axiomatizing quantum field theory, such as canonical quantization, path integrals, algebraic and topological descriptions, among others. 
In this section we will introduce the propagators relevant to our analysis and set the conventions in QFT that we will use throughout the manuscript.
%In this Section we will set the conventions in QFT that we will use throughout the manuscript. In this work, we will adopt the algebraic approach, which will allow us to conveniently describe the relevant propagators in the protocol of entanglement harvesting. 
We will focus on the theory of a real scalar quantum field in 3+1 Minkowski spacetime $\mathcal{M}$. %\shout{paragraph?} 

Consider a classical real scalar field $\phi(\mf x)$ satisfying the massless Klein-Gordon equation of motion
\begin{equation}
    \Box \phi = 0,
\end{equation}
where $\Box$ is the d'Alembertian. This equation uniquely defines the retarded Green's function $G_{\tc{r}}$ by the condition
\begin{equation}
   \Box G_{\tc{r}} f = f,
\end{equation}
for any source function $f\in C_0^\infty(\mathcal{M})$. The propagated function $G_\tc{r}f$ is supported in the causal future of the support of $f$.

We then promote $\phi(\xx)$ to an operator-valued distribution $\phi(\xx) \to \hat \phi(\xx)$, which we can expand in terms of creation and annihilation operators associated with the basis of plane wave solutions:
\begin{equation}
    \hat \phi(\xx) =  \frac{1}{(2\pi)^{\frac{3}{2}}}\int \frac{\dd^3\bm k}{\sqrt{2 |\bm k|}} \left( \hat a_{\bm k} \ee^{\ii \mf{k} \cdot \xx} + \hat a^\dagger_{\bm k} \ee^{-\ii \mf{k} \cdot \xx}\right),
\end{equation}
where $\mf k^0 \equiv |\bm k|$, $\mf k \cdot \xx = -|\bm k| t + \bm k \cdot \bm x$, and $(t,\bm x)$ is any set of inertial coordinates. The creation and annihilation operators satisfy the usual canonical commutation relations
\begin{equation}
    [\hat{a}_{\bm k},\hat{a}^\dagger_{\bm k'}] = \delta^{(3)} ( \bm k - \bm k'), \quad [\hat{a}_{\bm k},\hat{a}_{\bm k'}] = [\hat{a}^\dagger_{\bm k},\hat{a}^\dagger_{\bm k'}] = 0,
\end{equation} 
and uniquely define the Minkowski vacuum $\ket{0}$ by the condition $\hat{a}_{\bm k}\ket{0} = 0$ for all $\bm k$. Due to the fact that the field operator $\hat \phi(\xx)$ is only meaningful in the distributional sense, one introduces the (well defined) smeared field operator $\hat{\phi} : C_0^\infty(\mathcal{M}) \rightarrow C_0^\infty(\mathcal{M})$ as:
\begin{equation}
     \hat{\phi}(f) = \int \dd V f(\xx) \hat{\phi}(\xx),
\end{equation}
where $f$ is a test function whose profile gives the shape of the region where $\phi$ is smeared\footnote{One can also extend the action of the distribution $\hat{\phi}$ to more general spaces of test functions~\cite{Fredenhagen2015utr}.}. Operators of the form $\hat{\phi}(f)$ are the observables a local probe has access to when interacting with the field in the spacetime region defined by $|f|$.

% All the relevant observables of the theory can be written as products of these smeared field operators, thus a state is fully defined by the expectation values of products of the form $\hat{\phi}(f_1)...\hat{\phi}(f_n)$. 

 %Therefore, it is convenient to define the vacuum expectation value of the product of $n$ field operators, which we call the $n$-point function:
% \begin{equation}
%    \langle \hat{\phi}(f_1) \dots \hat{\phi}(f_n) \rangle \equiv \bra{0} \hat{\phi}(f_1) \dots \hat{\phi}(f_n) \ket{0}.
% \end{equation}

%1) We will focus on the vacuum

%2) Being a Gaussian state, it is fully determined by its two point function.

In this work, we will be interested in studying vacuum entanglement between finite spacetime regions, hence we will exclusively focus on the vacuum state $\ket{0}$. Being a Gaussian state, it is fully determined by its two-point function
\begin{equation}\label{eq:wightman-kernel}
    W(\xx, \xx') \!= \!\langle\hat{\phi}(\mf x) \hat{\phi}(\mf x')\rangle \!=\! \frac{1}{(2 \pi)^3} \!\int \!\dfrac{\dd^3 \bm k}{2 |\bm k|} \ee^{-\ii |\bm k| (t-t') + \ii \bm{k} \cdot (\bm x - \bm x')},
\end{equation}
where $\langle\hat{\mathcal{O}}\rangle \equiv \bra{0}\!\hat{\mathcal{O}}\!\ket{0}$ denotes the vacuum expectation value. When smeared against test functions, it defines the Wightman function
%In the case of free fields, the vacuum is a Gaussian state and all the odd $n$-point correlations functions vanish; the even $n$-point functions, on the other hand, reduce to the $2$-point function following Wick's theorem. Therefore, we can center our attention of the $2$-point function, which we call the Wightman function:
\begin{equation}
    W(f,g) = \langle \hat{\phi}(f) \hat{\phi}(g) \rangle   = \int \dd V \dd V' f(\xx) g(\xx') W(\xx, \xx').
\end{equation}
Moreover, we can use the decomposition $\hat{\phi}(\mf x) \hat{\phi}(\mf x') = \tfrac{1}{2}\{\hat{\phi}(\mf x),\hat{\phi}(\mf x')\} + \tfrac{\ii}{2} [\hat{\phi}(\mf x),\hat{\phi}(\mf x')]$ 
% \begin{equation}
%     \hat{\phi}(\mf x) \hat{\phi}(\mf x') = \tfrac{1}{2}\{\hat{\phi}(\mf x),\hat{\phi}(\mf x')\} + \tfrac{\ii}{2} [\hat{\phi}(\mf x),\hat{\phi}(\mf x')]
% \end{equation}
%can be used to 
to split the Wightman function into its symmetric and antisymmetric parts:
\begin{equation}\label{eq:wightman-function-H-E}
    W(f,g) \equiv \frac{1}{2} \bigg( H(f,g) + \ii E(f,g) \bigg),
\end{equation}
where $E$ is the (anti-symmetric) causal propagator and $H$ is the (symmetric) Hadamard distribution:
\begin{equation}\label{eq:hadamard-function}
    E(f,g) \equiv - \ii \langle [ \hat \phi(f), \hat \phi(g)]\rangle, \quad H(f,g) \equiv \langle \{ \hat \phi(f), \hat \phi(g)\}\rangle.
\end{equation}
%Notice that $H$ is symmetric while $E$ is anti-symmetric 
The kernel of the causal propagator can also be written in terms of the classical Green's functions $E(\mf x,\mf x') = G_\tc{r}(\mf x,\mf x') - G_\tc{r}(\mf x',\mf x)$, evidencing the fact that the antisymmetric part of the Wightman function is state independent. This implies that all the state dependence of the two-point function is encoded in the Hadamard term $H$. %This fact will be essential when discussing vacuum entanglement in QFT.
%and fully determined by the commutation relations of the theory. $H$ then contains all of the state dependence, while $E$ contains information about causal propagation within the theory. %Moreover, the Green's functions can also be used to define the causal propagator $E = G_\tc{r} - G_\tc{a}$, which maps compactly supported functions to solutions of the equations of motion $f\mapsto \phi = Ef$.
Another relevant propagator for describing interactions in QFT is the Feynman propagator $G_{\tc{f}}$, defined as the time-ordered two-point function:
\begin{equation}
    G_{\tc{f}}(\xx,\xx') = \theta(t-t') W(\xx,\xx') + \theta(t'-t) W(\xx',\xx),
\end{equation}
as well as its smeared version $G_\tc{f}(f,g)$. Here $\theta (t)$ is the Heaviside theta function, used to implement the time-ordering.
%  When smeared against test functions, it can be written as
% \begin{equation}
%     G_{\tc{f}}(f,g) \equiv \int \dd V \dd V' f(\xx) g(\xx') G_{\tc{f}}(\xx, \xx')
% \end{equation}
 %Using that the supports of the retarded Green's function is confined to the causal future, 
Similar to Eq.~\eqref{eq:wightman-function-H-E}, the Feynman propagator can be split as
\begin{equation}\label{eq:feynman-propagator}
    G_{\tc{f}}(f,g) = \frac{1}{2} \bigg( H(f,g) + \ii \Delta(f,g)\bigg),
\end{equation}
where we defined the symmetric propagator
\begin{equation}\label{eq:symmetric-propagator}
    \Delta(f,g) = G_{\tc{r}}(f,g)+G_{\tc{r}}(g,f).
\end{equation}
Similar to the decomposition of the Wightman function, all the state dependence of $G_\tc{f}$ is encoded in the Hadamard function, while $\Delta$ encodes the symmetric exchange of information between the supports of $f$ and $g$. %\\ \\ %sorry
The different propagator decompositions discussed here are particularly relevant for distinguishing genuine vacuum correlations from causal communication within the protocol of entanglement harvesting.

\section{Entanglement Harvesting}\label{sec:iii}

% \shout{start this part with smeared, but argue that the detectors tend to be much smaller than their light-crossing time and switch to pointlike on the example}

% \begin{itemize}
%     \item Brief review of Entanglement Harvesting (talk a bout detectors, etc.~\cite{Pozas-Kerstjens:2015,FullyRelativisticEH,closedform2024})
%     \item Mention how to quantify signalling vs. genuine entanglement extracted~\cite{ericksonNew,closedform2024}.
%     \item The example with the ``standard Gaussian" pointlike.
% \end{itemize}

In this section we will discuss the protocol of entanglement harvesting~\cite{Valentini1991,Reznik1,Pozas-Kerstjens:2015}, where two non-communicating probes can extract entanglement from a quantum field. We will review the protocol in Subsection~\ref{sub:reviewEH}, define a strict quantifier for the entanglement genuinely extracted from the field in Subsection~\ref{sub:strict}, and present an explicit example of entanglement harvesting with Gaussian time profiles in subsection~\ref{example-harvesting}.

\vspace{-4mm}

%\subsection{Setup}

%In the language of local algebras of Sec.~\ref{section:basics-aqft}, ideally, each probe will interact with independent localized algebras $\mathcal{A}(\mathcal{O}_\tc{a})$ and $\mathcal{A}(\mathcal{O}_\tc{b})$ in order to extract entanglement from the field degrees of freedom localized therein. 

\subsection{Review of entanglement harvesting}\label{sub:reviewEH}

The protocol of entanglement harvesting considers two localized probes that interact with a quantum field for finite times, in an attempt to extract the entanglement between the regions that the detectors couple to. When the probes become entangled after coupling to the field, one can only claim that entanglement between the probes was extracted from the field when communication between them is negligible~\cite{ericksonNew,quantClass}. Thus, one typically considers spacelike separated interactions, which can then be used to infer entanglement between independent field degrees of freedom.  

The standard formulation of the protocol describes the probes as particle detectors, according to the Unruh-DeWitt (UDW) model~\cite{Unruh1976,DeWitt}. This model has been shown to accurately describe interactions between atoms and an external electromagnetic field~\cite{Pozas2016,Nicho1,richard}, nucleons and the neutrino fields~\cite{neutrinos,carol}, general systems with gravity~\cite{pitelli,boris}, and can model general non-relativistic quantum systems~\cite{Unruh-Wald,generalPD}, as well as localized quantum fields~\cite{Unruh-Wald,QFTPD}. Moreover, these models have a wide range of application in relativistic quantum information protocols, ranging from quantum energy teleportation~\cite{teleportation,Hotta2011,HottaDistance,nichoTeleport,borisTeleportRev} and quantum collect calling~\cite{Jonsson2,PRLHyugens2015,collectCalling}, to quantum computing~\cite{Martin-Martinez_Aasen_Kempf_2013,Martin-Martinez_Sutherland_2014,martin-martinez2015,Layden_Martin-Martinez_Kempf_2016,phil} and, more relevant to this work, entanglement harvesting~\cite{Valentini1991,Reznik1,reznik2,NickEdu2014,Pozas-Kerstjens:2015,Pozas2016,EricksonZero,HarvestingSuperposed,Cong2019,ericksonNew,carol,HarvestingAccelerationRobb,HarvestingBHLaura,HarvestingDelocalized,HarvestingQueNemLouko,adrianQuenchedHarv2025,hectorMass,boris,FullyRelativisticEH}. 

For simplicity, we will focus on the case where the probes are two-level UDW detectors undergoing comoving inertial motion in Minkowski spacetime. %These detectors provide accurate results whenever only two energy levels are relevant in the interaction. 
Each detector (labelled A and B) can therefore be thought of as a qubit traveling along a timelike trajectory $\mf z_i(t)$ parametrized by their common proper inertial time $t$ for $i\in \{\tc{A},\tc{B}\}$. Their respective free Hamiltonians generating evolution with respect to $t$ are given by

\vspace{-4mm}

\begin{equation}
    \hat{H}_i = \Omega \, \, \hat{\sigma}_i^+ \hat{\sigma}_i^-,
\end{equation}

\vspace{0.5mm}

\noindent where $\hat{\sigma}_i^\pm$ are the two-level raising and lowering operators acting on the ground and excited states $\ket{g_i}$ and $\ket{e_i}$, and $\Omega \geq 0$ is their energy gap.

The detectors interact with the field in distinct regions of spacetime defined by spacetime smearing functions $\Lambda_i(\mf x)$, which we assume to be localized in both space and time. The interaction can then be described by the local Hamiltonian densities:
\begin{equation}\label{eq:interaction-hamiltonian}\
    \hat{h}_{I, i}(\xx) = \lambda \, \big( \Lambda^+_i(\xx) \hat{\sigma}_i^+ +\Lambda^-_i(\xx) \hat{\sigma}_i^-\big)\, \hat{\phi}(\xx), 
\end{equation}
where $\lambda$ is a dimensionless coupling constant and we defined $\Lambda^\pm_i(\xx) = \Lambda_i(\xx) \, \ee^{\pm\ii\Omega  t}$. It is also typical to assume that the smearing functions factor as $\Lambda_i(\xx) \equiv \chi_i(t) F_i(\bm{x})$, where $\chi_i(t)$ are switching functions that control the temporal profile of the interactions and $F_i(\bm x)$ are the spatial smearing functions of the detectors~\cite{eduardo,us}. For instance, for the protocol of entanglement harvesting, one could consider sufficiently separated $F_\tc{a}(\bm x)$ and $F_\tc{b}(\bm x)$ with $\chi_\tc{a}(t) = \chi_\tc{b}(t)$, resulting in effectively spacelike separated detectors.

%Effectively, each detector couples to the algebra of observables in $\mathcal{O}_i$, defined as the causal hull of $\text{supp}(\Lambda_i)$. This allows the detectors to locally probe the field degrees of freedom in $\mathcal{A}(\mathcal{O}_\tc{a})$ and $\mathcal{A}(\mathcal{O}_\tc{b})$, as well as correlations between these regions. 
For a concrete example, let us assume that the initial state of each detector is $\hat{\rho}_{i, 0} = \hat{\sigma}_i^- \hat{\sigma}_i^+ = \ket{g_i} \!\bra{g_i}$ and that the field is initially in the vacuum state $\ket{0}$. 
The initial detectors-field state is then given by 

\vspace{-4mm}

\begin{equation}
    \hat{\rho}_0 = \hat{\rho}_{\tc{d}, 0} \otimes \ket{0}\!\!\bra{0}, \,\,\,\, \hat{\rho}_{\tc{d}, 0} = \hat{\rho}_{\tc{a}, 0} \otimes \hat{\rho}_{\tc{b}, 0}.
\end{equation}

\vspace{0.5mm}

This state evolves to a final state $\hat{\rho} = \hat{U}_I  \hat{\rho}_0\hat{U}_I^\dagger$, where $\hat{U}_I$ is the time evolution operator

\vspace{-4mm}

\begin{equation}\label{eq:interaction-hamiltonian}
    \hat{U}_I =  \mathcal{T}_t \,\text{exp}\Big(-\ii \!\int\! \dd V \hat{h}_I(\xx)\Big),
\end{equation}

\vspace{0.5mm}

\noindent where
\begin{equation}
    \hat{h}_I (\xx) = \hat{h}_{I, \tc{a}}(\xx) + \hat{h}_{I, \tc{b}}(\xx),
\end{equation}

\vspace{0.5mm}

\noindent and $\mathcal{T}_t\exp$ denotes the time ordered exponential with respect to the time parameter $t$ (see~\cite{us2} for details).  To leading order in  $\lambda$, the final state of the detectors, given by $\hat{\rho}_\tc{d} = \Tr_\phi ( \hat{U}_I \hat{\rho}_0 \hat{U}_I^\dagger)$, can be written (in the $\{ \ket{g_\tc{a}g_\tc{b}}, \ket{g_\tc{a}e_\tc{b}}, \ket{e_\tc{a} g_\tc{b}}, \ket{e_\tc{a}e_\tc{b}} \}$ basis) as

\vspace{-2.5mm}

\begin{equation}\label{eq:rhodmatrix}
    \hat{\rho}_\tc{d} = \begin{pmatrix}
        1 -  W_{\tc{a}\tc{a}}^{\sct{-+}} - W_{\tc{b}\tc{b}}^{\sct{-+}}  & 0 & 0 & -( G_{\tc{a}\tc{b}}^{\sct{++}})^* \\
        0 & W_{\tc{b}\tc{b}}^{\sct{-+}} & W_{\tc{a}\tc{b}}^{\sct{-+}} & 0  \\
        0& W_{\tc{b}\tc{a}}^{\sct{-+}} & W_{\tc{a}\tc{a}}^{\sct{-+}} &0 \\
       -G_{\tc{a}\tc{b}}^{\sct{++}} &0&0&0
     \end{pmatrix} + \mathcal{O}(\lambda^4),
\end{equation}

\vspace{0.5mm}

\noindent where we denote

\vspace{-4mm}

\begin{equation}
    G_{i\,j}^{\sct{++}} = \lambda^2 G_{\tc{f}}(\Lambda_{i}^+,\Lambda_{j}^+), \,\,\,\, W_{\,i\,i}^{\sct{-+}} = \lambda^2 W(\Lambda_{i}^-,\Lambda_{i}^+).
\end{equation}

\vspace{0.5mm}

\noindent Notice that this leading-order description is accurate whenever the second order terms are sufficiently small to allow for a perturbative treatment, and the higher order corrections can also be neglected. We assume $\lambda$ to be arbitrarily small through our theoretical description, and in Section~\ref{sec:limit} we will discuss cases where this perturbative treatment can accurately describe physical systems, as well as the cases where it fails. See~\cite{Pozas-Kerstjens:2015} for explicit computations and typical orders of magnitude of the terms of order $\lambda^4$. 

One can quantify the leading order entanglement between the detectors through the negativity, a faithful entanglement quantifier for a system of two qubits~\cite{VidalNegativity,plenio}. This quantifier is inspired by Peres-Horodecki's positive partial transpose (PPT) criterion~\cite{Peres_1996,Horodecki_1996}, which states that $\hat{\rho}_\tc{d}$ is separable if its partial transpose is a positive semi-definite matrix, i.e. if all its eigenvalues are positive. In other words, if at least one eigenvalue of $\ptrans{\hat{\rho}}_\tc{d}$ is negative, the state $\hat{\rho}_\tc{d}$ is guaranteed to be entangled. The negativity $\mathcal{N}$ quantifies how negative these eigenvalues are through the expression

\vspace{-4mm}

\begin{equation}
     \mathcal{N}(\hat{\rho}_\tc{d}) = \sum_{\gamma_i < 0} |\gamma_i|,
\end{equation}

\vspace{0.5mm}

\noindent where $\gamma_i$ are the eigenvalues of  $\ptrans{\hat{\rho}}_\tc{d}$. If $ \mathcal{N}(\hat{\rho}_\tc{d})=0$, the state is separable, otherwise, if $ \mathcal{N}(\hat{\rho}_\tc{d})>0$, it is entangled, and the specific value of $\mathcal{N}$ tells us by how much the PPT criterion is violated.
To leading order in $\lambda$, $\ptrans{\hat{\rho}}_\tc{d}$ has only one potentially negative eigenvalue

\vspace{-2mm}

\begin{equation}\label{eq:gamma-}
    \gamma_- = {{\sqrt{|G_{\tc{a}\tc{b}}^{\sct{++}}|^2+\left(\frac{ W_{\tc{a}\tc{a}}^{\sct{-+}}-W_{\tc{b}\tc{b}}^{\sct{-+}}}{2}\right)^2 } - \frac{W_{\tc{a}\tc{a}}^{\sct{-+}} + W_{\tc{b}\tc{b}}^{\sct{-+}}}{2}}},
\end{equation}

\vspace{1.5mm}

\noindent so that $\mathcal{N}(\hat{\rho}_\tc{d}) = \text{max} \left( 0, \gamma_- \right)  + \mathcal{O}(\lambda^4)$. In the case of identical inertial comoving detectors, which we will assume from this point on, the negativity simply reduces to

\vspace{-4mm}

\begin{equation}\label{eq:negground}
     \mathcal{N}(\hat{\rho}_\tc{d}) = \text{max} \Big( \,0\,, |G_{\tc{a}\tc{b}}^{\sct{++}}| - W^{\sct{-+}} \Big)  \\
    + \mathcal{O}(\lambda^4),
\end{equation}

\vspace{0.5mm}

\noindent where $W^{\sct{-+}} \equiv W_{\tc{a}\tc{a}}^{\sct{-+}} = W_{\tc{b}\tc{b}}^{\sct{-+}}$. 

\vspace{1mm}

Notice that, under the assumption that $\lambda$ is an infinitesimal parameter, the fourth order corrections to the negativity are negligible compared to the terms displayed in Eq.~\eqref{eq:negground}. However, given that the negativity involves a subtraction of second order terms, it is naturally smaller than the second order corrections to $\hat{\rho}_\tc{d}$ in Eq.~\eqref{eq:rhodmatrix}, and when considering finite (small) values for the coupling constant, the fourth order corrections might affect $\mathcal{N}(\hat{\rho}_\tc{d})$ non-trivially. This will be at the core of the discussion in Section~\ref{sec:limit} on the limitations of second order perturbation theory to describe entanglement harvesting protocols.

\vspace{1mm}

For the negativity to be non-zero, the non-local term $G_{\tc{a}\tc{b}}^{\sct{++}}$, encoding non-local field correlations, has to be larger than the local terms $W_{\tc{a}\tc{a}}^{\sct{-+}}$, $W_{\tc{b}\tc{b}}^{\sct{-+}}$, representing the detectors' acquired noise due to local vacuum fluctuations.
%The non-local term is what produces the entanglement: it quantifies the probability amplitude that the quantum field simultaneously excited both detectors, and it arises from the fact that both detectors are interacting with the same quantum field.
%On the other hand, the local terms represent the probability that each detector becomes excited due to its interaction with local vacuum fluctuations of the field. 
This local noise is also the excitation probability of a single detector:

\vspace{-4mm}

\begin{equation}
    \Tr\big(\hat{\rho}_\tc{d} \ket{e_i}\!\! \bra{e_i}\big) = W_\tc{\:\!i\,i}^{\sct{-+}},
\end{equation}

\vspace{0.5mm}

which also determines the leading order entanglement of each detector with the field.

\vspace{1mm}

The negativity of Eq.\eqref{eq:negground} then quantifies the total amount of entanglement acquired by the probes after their interaction with the field. Importantly, it alone is not enough to analize whether this entanglement was genuinely harvested from the field, as we will discuss in the following subsection.

\subsection{Quantifying genuine entanglement harvesting}\label{sub:strict}
\vspace{2mm}
Entanglement between the detectors alone is not enough to conclude that there was pre-existing vacuum entanglement between the regions they couple to. Indeed, there are two ways in which the detectors can become entangled, and only one of them corresponds to genuine entanglement harvesting~\cite{ericksonNew}. %\\ \\%sorry
The first possibility is that the detectors become entangled by sharing quantum information through the field, in which case the field merely acts as a mediator, effectively producing a direct retarded coupling between the probes~\cite{quantClass}. %\\%sorry
The other possibility is when the interaction regions are effectively causally disconnected, avoiding field-mediated communication. In this case, the entanglement between the detectors arises solely from pre-existing entanglement in the field. This is the case where the protocol of entanglement harvesting takes place.

%Unfortunately, compact support is not realizable in practice, 
Although, ideally, one would always consider spacelike separated interaction regions in entanglement harvesting, these idealized regions are defined by compactly supported spacetime smearing functions that can be hard to implement in practice. Therefore, in most cases, there will be some small level of communication between the detectors. This raises the question of how to quantify signalling and genuinely harvested entanglement. One way of addressing this questions involves separating the smeared Feynman propagator as in Eq.~\eqref{eq:feynman-propagator},

\vspace{-4mm}

\begin{equation}\label{eq:GABHABDELTAAB}
    G_{\tc{ab}}^{\sct{++}} = \tfrac{1}{2}H_{\tc{Ab}}^{\sct{++}} + \tfrac{\ii}{2}\Delta_{\tc{ab}}^{\sct{++}},
\end{equation}

\vspace{0.5mm}

\noindent where we defined $H_{\tc{ab}}^{\sct{++}} = \lambda^2 H (\Lambda_{\tc{a}}^+,\Lambda_{\tc{b}}^+ )$  and $\Delta_{\tc{ab}}^{\sct{++}} = \lambda^2 \Delta (\Lambda_{\tc{a}}^+,\Lambda_{\tc{b}}^+)$. In this case, the Hadamard function $H_{\tc{ab}}^{\sct{++}}$ encodes the state-dependent correlations between the regions, while $\Delta_{\tc{ab}}^{\sct{++}}$ encodes the symmetric exchange of information between the detectors\footnote{For instance, two spacelike separated interaction regions imply $\Delta_{\tc{ab}}^{\sct{++}} = 0$.}. Therefore, if one wants to ensure that the detectors effectively cannot communicate, one must consider situations where $\Delta_{\tc{ab}}^{\sct{++}}$ is negligible. 

Concretely, the condition for genuine harvesting was formulated as $\frac{1}{2} |\Delta_{\tc{ab}}^{\sct{++}}| \ll \mathcal{N}(\hat{\rho}_\tc{d})$ in~\cite{closedform2024}. 
More generally, using the fact that the Hadamard propagator encodes the state dependence of the field, as well as the genuine vacuum effects in QFT, we can introduce a quantifier for genuine entanglement extraction.% \\%sorry
We define the \textit{signalling-to-entanglement ratio} (SER) as the ratio between the negativity when neglecting all local smeared propagators and the non-local state-dependent terms, denoted by ${\mathcal{N}^\circ}(\hat{\rho}_\tc{d})$\footnote{Explicitly, in the expression for $\mathcal{N}(\hat{\rho}_\tc{d})$, terms that involve two-point functions smeared against the same detector (such as  $W_{\,i\,i}^{\sct{-+}}$ in Eq.~\eqref{eq:rhodmatrix}) and the state dependent terms (involving field anti-commutators) smeared against functions associated to the different detectors A and B ($H_{\tc{ab}}^{++}$ in Eq.~\eqref{eq:GABHABDELTAAB}).}, and the total negativity:
\begin{equation}\label{eq:SNR}
     \Theta(\hat{\rho}_\tc{d}) = \begin{cases} \displaystyle{\frac{\mathcal{N}^\circ(\hat{\rho}_\tc{d})}{\mathcal{N}(\hat{\rho}_\tc{d})}}, & \text{if} \,\,\,\mathcal{N}(\hat{\rho}_\tc{d}) \neq 0 ,\\
     \,\,\,0\,, &\text{otherwise.}
     \end{cases}
\end{equation}

Notice that in Eq.~\eqref{eq:gamma-}, when setting the local terms $W_{\,i\,i}^{\sct{-+}}$ and the non-local anti-commutator term $H_{\tc{ab}}^{\sct{++}}$ to zero, we find $\mathcal{N}^\circ(\hat{\rho}_\tc{d})= \tfrac{1}{2}|\Delta_{\tc{ab}}^{\sct{++}}|$. This ensures that when the detectors start in the ground state, the condition $\Theta(\hat{\rho}_\tc{d})\ll 1$ is equivalent to \mbox{$\frac{1}{2} |\Delta_{\tc{ab}}^{\sct{++}}| \ll \mathcal{N}(\hat{\rho}_\tc{d})$}. Overall, $\Theta \sim 1$ implies that most of the entanglement acquired by the detectors is due to communication, while $\Theta \sim 0$ characterizes genuine entanglement harvesting.

It is important to compare the SER defined above with the other standard quantifier for genuine harvested entanglement, the communication-mediated entanglement estimator (CMEE), introduced in~\cite{ericksonNew}. In essence, the SER is a more strict estimator of genuine entanglement harvesting than the CMEE. Indeed, the CMEE only sets the non-local Hadamard terms ($H(\Lambda_\tc{a}^\pm,\Lambda_\tc{b}^{\pm})$) to zero, while keeping the local noise terms ($H(\Lambda_i^\pm,\Lambda_i^{\pm})$). For instance, when the detectors start their interaction in the ground state, this would be equivalent to keeping the term $-(W_\tc{aa}^{\sct{-+}} + W_\tc{bb}^{\sct{-+}})$ in Eq.~\eqref{eq:gamma-}, decreasing the overall value of the numerator in Eq.~\eqref{eq:SNR}. Moreover, the definition used in~\cite{ericksonNew}, relies on the specific form of the negativity in Eq.~\eqref{eq:negground}, which is only valid when the detectors' local noise is identical, and they are prepared in the ground state. On the other hand, the definition of the signalling-to-entanglement ratio $\Theta(\hat{\rho}_\tc{d})$ is valid for general initial states of the detectors, general interaction regions  and general energy gaps. A detailed comparison between the SER and CMEE estimators in different cases will be presented in~\cite{SERCMEE}.

%$ \mathcal{N}(\hat{\rho}_\tc{d}) > 0$ and .
%In other words, we want the negativity to be non-zero while making sure that this stems from non-local correlation, and not the signalling terms ($\Delta$).
%\subsection{Correlations}
%The main objective of this work is to characterize the genuine entanglement present in the vacuum between two causally disconnected regions of spacetime.
%In the protocol of entanglement harvesting, one picks interaction regions $\Lambda_i(\mf x)$ 
%places each UDW detector in a region and let them interact locally with the quantum field vacuum for a finite period of time (this is of extreme importance, since \tbb{vacuum fluctuations are asymptotically zero}, and it is enforced by means of the switching function $\chi(t)$).
\subsection{Canonical example}\label{example-harvesting}

We will now discuss a concrete example of entanglement harvesting from the Minkowski vacuum using Gaussian spacetime smearing functions. We further assume that the detectors' interactions are switched on and off simultaneously in their comoving frame ($\chi_\tc{a}(t) = \chi_\tc{b}(t) = \chi(t)$), and that their shape is identical and separated by $\bm L$ ($F_\tc{b}(\bm x) = F_\tc{a}(\bm x - \bm L)$). Explicitly,
% %  We also assume that the interaction regions factor as
% % \begin{equation}\label{eq:smearing-function}
% %     \Lambda_i(\xx) \equiv \chi_i(t) F_i(\bm{x}),
% % \end{equation}
% % where $F_i(\bm{x})$ defines the spatial smearing  of each detector and $\chi_i(t)$ are the switching functions, deter   mining the interaction profile in time. Equation~\ref{eq:smearing-function} corresponds to the assumption of rigid detectors~\cite{eduardo,us}. 
% For an explicit example, we consider the spacetime smearing functions
\begin{align}
    \Lambda_\tc{a}(\xx) &= \dfrac{1}{(2 \pi \sigma^2)^\frac{3}{2}} \ee^{-\frac{t^2}{2T_0^2}}  \, \ee^{-\frac{|\bm{x}|^2}{2\sigma^2}}, \\
    \Lambda_\tc{b}(\xx) &= \dfrac{1}{(2 \pi \sigma^2)^\frac{3}{2}} \ee^{-\frac{t^2}{2T_0^2}}  \, \ee^{-\frac{|\bm{x}-\bm{L}|^2}{2\sigma^2}},
\end{align}
where the parameters $\sigma$ and $T_0$ control the spatial width and time duration of the interaction. Moreover, assuming that the the detectors' light-crossing time is much smaller than the interaction time $T_0$, we have $\sigma\ll T_0$. In this regime, one can effectively consider the pointlike limit $\sigma\to 0$, where $\Lambda_\tc{a}(\mf x) \to \ee^{- t^2/2T_0^2} \delta^{(3)}(\bm x)$ and \mbox{$\Lambda_\tc{b}(\mf x) \to \ee^{- t^2/2T_0^2} \delta^{(3)}(\bm x - \bm L)$}. For these parameter choices, closed-form expressions for $G_\tc{ab}^{\sct{++}}$ and $W_\tc{aa}^{\sct{-+}} = W_\tc{bb}^{\sct{-+}}$ were found in~\cite{Ng1,closedform2024}.
% :
% \begin{equation}
%     G_{\tc{a}\tc{b}}^{\sct{++}} = \frac{\lambda^2 T}{4 L \sqrt{\pi}} \ee^{-\frac{L^2}{4 T^2}-\Omega^2 T^2} \left(\text{erfi} \left( \frac{L}{2 T} \right) -\ii \right),
% \end{equation}
% \begin{equation}
%     W_{\tc{a}\tc{a}}^{\sct{-+}}=W_{\tc{b}\tc{b}}^{\sct{-+}} = \frac{\lambda^2\ee^{-\Omega^2T^2}}{4 \pi} \left(1-\sqrt{\pi} \,\Omega T \ee^{\Omega^2T^2}\text{erfc}(\Omega T)\right),
% \end{equation}
% and hence, the negativity can be expressed as:

% \begin{widetext}
% \begin{equation}
% \mathcal{N}(\hat{\rho}_{\tc{a}\tc{b}}) = \text{max} \left\{0, \frac{\lambda^2 \ee^{- \Omega^2 T^2}}{4\pi}\left(\sqrt{\pi} \ee^{-\frac{L^2}{4T^2}}\frac{T}{L}\sqrt{1 + \text{erfi} \left( \frac{L}{2 T} \right)^2}\right) + \sqrt{\pi} \,\Omega T \ee^{\Omega^2T^2}\text{erfc}(\Omega T) - 1\right\} + \mathcal{O}(\lambda^4).
% \end{equation}
% \end{widetext}
% Moreover, the signalling estimator can also be found in closed form:
% \begin{equation}
%  \tfrac{1}{2}\Delta^{\sct{++}}_{\tc{ab}} = - \frac{\lambda^2 T}{4 L \sqrt{\pi}}    \ee^{-\frac{L^2}{4 T^2}-\Omega^2 T^2} .
% \end{equation}

\begin{figure}[h!]
    \centering
    \includegraphics[width = \linewidth]{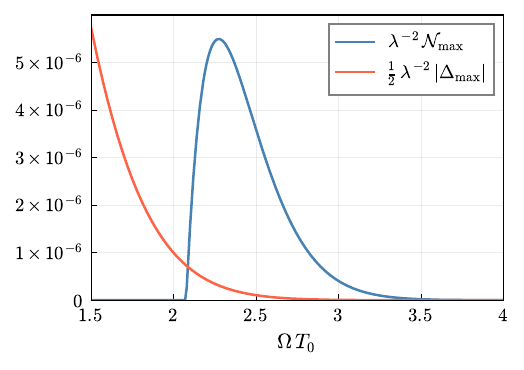}
    \vspace*{-10mm}
    \caption{Negativity and signalling contribution for two identical UDW point-like detectors, interacting with a real massless scalar field in Gaussian spacetime regions separated by $|\bm L|=5 T_0$.
    }
    \label{figure:example-gauss}
\end{figure}

In Fig.~\ref{figure:example-gauss}, we plot the leading order negativity and signalling of the two detectors as a function of $\Omega T_0$ when $|\bm L| = 5 T_0$. Notice that the values displayed in the plot are all scaled by $\lambda^{-2}$, as the leading order corrections are proportional to $\lambda^2$. In this example, we can see values of $\Omega$ such that the condition for genuine harvesting $\Theta(\hat{\rho}_\tc{d}) \ll 1$ is approximately satisfied, as we have $\frac{1}{2} |\Delta_{\tc{ab}}^{\sct{++}}| \ll \mathcal{N}(\hat{\rho}_\tc{d})$. For instance, at the peak, $\Omega T_0 \approx 2.3$, signalling amounts to approximately $5\%$ of the negativity. %\\ \\ %sorry
This means that if a setup of this sort is implemented in practice, one would be able to infer vacuum entanglement between the regions the detectors couple to. However, it would be natural to question whether the small order of magnitude of $~10^{-6} \lambda^2$ would allow any phenomena of this type to be observed in practice. In the remaining of this manuscript we will find optimal switching functions for the protocol, which will improve this result by several orders of magnitude. %While for a the given switching function $\chi(t)$, it is possible to optimize the adimensional parameters $(\Omega T, L/T)$ that define this setup, ideally, one would also optimize over all possible choices of $\chi(t)$. %However, such an optimization would be numerically very costly, given the lack of closed-form expressions for the setup and the oscillatory nature of the integrals defining $G_{\tc{a}\tc{b}}^{\sct{++}}$.

\newcommand{\lone}{\vspace{-4mm}}

\newcommand{\ltwo}{\vspace{0.5mm}}

\section{Computing Entanglement for Arbitrary Time Profiles}\label{sec:iv}

In this section we will describe a method through which one can obtain closed-form expressions for the negativity and signalling estimator in entanglement harvesting for arbitrary temporal profiles in 3+1 Minkowski spacetime. Our method expands a given switching function in a basis of $L^2(\mathbb{R})$ and computes $G_\tc{ab}^{\sct{++}}$, $W_\tc{aa}^{\sct{-+}}$, $W_\tc{bb}^{\sct{-+}}$ and $\Delta_\tc{ab}^{\sct{++}}$ through generating functions.

Throughout this section we will consider spacetime smearing functions of the form

\lone

\begin{align}\label{eq:arbitrary-smearing-function}
   & \Lambda_{\tc{a}}(\xx) \equiv \chi(t) F(\bm x), \quad \Lambda_{\tc{b}}(\xx) \equiv \chi(t) F(\bm x - \bm L),
\end{align}

\ltwo

\noindent where we assume that $\chi$ and $F$ are rapidly decreasing $L^2$ functions. We then consider a basis of real smooth functions $\{\chi_n\}$ of $L^2(\mathbb{R})$, so that we can expand $\chi$ as

\lone

\begin{equation}\label{eq:chi}
   \chi(t) = \smash{\sum_{n=0}^{\infty}} c_n \chi_n(t),\phantom{\Big|} 
\end{equation}

\ltwo

\noindent where $c_n = \langle \chi_n,\chi\rangle_{L^2}$ %\trr{and $T_0$ is a parameter with units of time that ensures that the components $c_n$ are dimensionless}
. Using this expression in Eq.~\eqref{eq:arbitrary-smearing-function}, we have

\lone

\begin{align}\label{eq:lambda-ring}
     \Lambda_\tc{a}(\xx) \equiv \smash{\sum_{n=0}^\infty} c_n {\Lambda}_{\tc{a},n}(\xx),\phantom{\Big|}  
\end{align}

\ltwo

\noindent where ${\Lambda}_{\tc{a},n}(\mf x) = \chi_n(t)F(\bm x)$, and similarly for $\Lambda_\tc{b}(\mf x) = \chi_n(t)F(\bm x - \bm L)$. 

We can now use this basis expansion to compute the relevant propagators in the protocol of entanglement harvesting. Indeed, define

\lone

\begin{equation}\label{eq:lambda-ring-pm}
    {\Lambda}_{i,n}^\pm(\mf x) = \ee^{\pm \ii \Omega t}{\Lambda}_{i,n}(\mf x),
\end{equation}

\ltwo

\noindent and let $P$ be a general propagator (e.g. $G_\tc{f}$, $W$, $H$, or $\Delta$). We can then write

\lone

\begin{align}\label{eq:coefficients-polynomial-development}
    P(\Lambda_{\tc{a}}^\pm,\Lambda_{\tc{b}}^\pm)=\smash{\sum_{n=0}^\infty \sum_{m=0}^\infty} c_n  c_m \, \underbrace{P ({\Lambda}^\pm_{\tc{a}, n},{\Lambda}^\pm_{\tc{b}, m})}_{\displaystyle{P_{nm}}}.  \phantom{\Big|}
\end{align}

\ltwo

\noindent
For convenience, we now truncate the basis $\{\chi_n\}$ for $n \leq N$, obtaining a partial description of our switching functions that can be made arbitrarily precise in the limit $N\to\infty$. This allows us to write the smeared propagator in terms of matrix products:

\lone

\begin{equation}\label{eq:matrix-product}
    P(\Lambda_{\tc{a}}^\pm,\Lambda_{\tc{b}}^\pm) \approx \trans{\bm c} \mf P \,\bm c,
\end{equation}

\ltwo

\noindent
where $\bm c$ is the $N$ dimensional vector with components $c_n$ and $\mf P$ is the $N\times N$ matrix with components $P_{nm}$. 
Equation~\eqref{eq:matrix-product} effectively replaces the computation of non-trivial multidimensional integrals by simple matrix products. However, it requires one to have precomputed all the coefficients $c_n$, as well as the more challenging matrix elements $P_{nm}$. Although the $c_n$ are given by one dimensional integrals and can usually be computed efficiently, this is not always the case for $P_{nm}$. However, by making specific choices for both $F(\bm x)$ and the basis $\chi_n(t)$, it is possible to find closed-form expressions for the components of the matrix $\mf P$. Indeed, by picking

\lone

\begin{equation}\label{eq:choice-F}
    F(\bm x) = \frac{\ee^{- \frac{|\bm x|^2}{2\sigma^2}}}{(2\pi \sigma^2)^{3/2}},
\end{equation}

\ltwo

\noindent

and the Hermite basis $\{h_n(t,T)\}$,  given by
\begin{equation}\label{eq:h(t)}
    h_n(t,T) =\frac{\pi^{-1/4}}{\sqrt{2^n n!\, T}} H_n(t/T) \, \ee^{-\frac{t^2}{2T^2}},
\end{equation}

\ltwo

\noindent
we can find closed-form expressions for $P_{nm}$ in terms of a generating function $P(\alpha,\beta)$. In fact, in Appendix~\ref{app:matrix}, we find closed-form expressions for the generator

\lone

\begin{equation}
    P(\alpha,\beta) = P (\Lambda^\pm_{\alpha}, \Lambda^\pm_{\beta})
\end{equation}

\ltwo

\noindent
in the pointlike limit $\sigma\to 0$, where

\lone

\begin{align}
    \Lambda^\pm_{\alpha}(\mf x) &= \ee^{\alpha t}\ee^{\pm \ii \Omega t} \ee^{- \frac{t^2}{2T^2}} \delta^{(3)}(\bm x),\\ \nonumber\\
    \Lambda^\pm_{\beta}(\mf x') &= \ee^{\beta t'}\ee^{\pm \ii \Omega t'} \ee^{- \frac{t'{}^2}{2T^2}} \delta^{(3)}(\bm x' - \bm L).
\end{align}

\ltwo

\noindent
Noticing that $t^n = \left.\frac{\dd^n}{\dd \alpha^n}\ee^{\alpha  t}\right|_{\alpha=0}$, the fact that

\lone

\begin{equation}\label{eq:feynman-trick}
     H_n (t/T) = \left.\bigg[ H_n \bigg( \frac{1}{T}\dfrac{\dd}{\dd \alpha} \bigg) \ee^{\alpha t}\bigg]\right|_{\alpha = 0}
\end{equation}

\ltwo

\noindent
then allows us to write (see Appendix~\ref{app:matrix} for details)

\lone

\begin{align}\label{eq:Pnm}
    P_{nm} &= \Bigg[\frac{\pi^{-1/2}}{\sqrt{2^{n+m} n! \, m! } \,T} \\
    &\quad \quad \times H_n \bigg( \frac{1}{T}\dfrac{\dd}{\dd \alpha} \bigg)H_m \bigg( \frac{1}{T}\dfrac{\dd}{\dd \beta} \bigg) P(\alpha,\beta)\Bigg]\Bigg\lvert_{\alpha=\beta=0}.\nonumber
\end{align}

\ltwo

\noindent
One can also evaluate the smeared local propagators $P(\Lambda_\tc{a}^\pm,\Lambda_\tc{a}^\pm)$ and $P(\Lambda_\tc{b}^\pm,\Lambda_\tc{b}^\pm)$ by taking $\bm L = 0$. %\\ \\ %sorry 
With these expressions, computing the matrices is just a matter of taking the different derivatives of the generator $P(\alpha,\beta)$ with respect to the parameters $\alpha$ and $\beta$%\footnote{Although this can become computationally expensive, these derivatives only need to be computed once, and can then be reused for any switching functions.}
.  Once the corresponding matrices for the Feynman propagator ($\mf G$) and Wightman function ($\mf W$) are computed (see Appendix~\ref{app:leibniz} for explicit expressions), the negativity can be expressed in terms of simple matrix products:
%\begin{equation}\label{eq:concurrence}
%    \mathcal{N}(\bm a, \bm b) = |\trans{\bm{a}}\, G \bm b| - \sqrt{(\trans{\bm{a}} W \bm a) %(\trans{\bm{b}} W \bm b)},
%\end{equation} where $G$ and $W$ are defined in Eqs.~\eqref{eq:Gnm}-\eqref{eq:Wnm}.

\lone

\begin{equation}\label{eq:concurrence}
    \mathcal{N}(\bm c) = |\trans{\bm{c}} \mf G \,\bm c| - \trans{\bm{c}} \mf W \,\bm c,
\end{equation}

\ltwo

\noindent
whenever it is non-zero. Notice that the matrix $\mf G$ is symmetric with complex components and $\mf W$ is Hermitian, due to the facts that $G_\tc{f}(f,g) = G_\tc{f}(g,f)$ and $W(f,g) = W(g^*,f^*)^*$. Analogously, the matrices corresponding to the Hadamard and symmetric propagators $\mf H$ and $\mf \Delta$ are both symmetric with complex components. %\\ \\ %sorry
Using this technique, we developed an online tool that expands an arbitrary switching function and computes the relevant smeared propagators within fractions of seconds~\cite{EHEapp}, allowing the user to export the generated data. 

%\shout{tease spatial smearing?} \trr{This method can be applied to the spatial smearing function. In App.~\ref{}, we also give closed-form expressions for the propagators using a fixed Gaussian switching function and arbitrary spatial shape. As it turns out, in most scenarios the spatial shape of the detectors is fixed (\cite{perche2026bosepolaronsrelativisticunruhdewitt}) and the switching function is the part that contributes the most. Therefore, we will center our optimization on the switching function.}

\section{Optimizing Entanglement Harvesting}\label{sec:optimization}

In this section, we will optimize the protocol of entanglement harvesting by finding switching functions that maximize the leading order extracted negativity when pointlike detectors are at different levels of causal separation. However, we note that this problem is not yet well posed in the context of Sections~\ref{sec:iii} and~\ref{sec:iv}, as we have an arbitrary coupling constant, and scaling the switching function could increase the negativity arbitrarily. Therefore, we must first find a criterion to compare the negativity that is invariant under a rescaling of the total coupling constant $\lambda \chi(t)$. Ideally, we want $\text{max}_t|\chi(t)|$ to be approximately the same for all switching functions involved. If there exists a fixed interval where $\chi(t)$ is supported (or effectively supported), then this condition can be achieved by normalizing our switching functions with respect to any norm. For convenience, we will choose the $L^2$ norm, where
\begin{equation}
    \chi(t) = \smash{\sum_{n=0}^N} c_n \chi_n(t)\quad \Longrightarrow \quad||\chi||^2 = \trans{\bm c}\bm c.\phantom{\Big|}
\end{equation}
To ensure a fair comparison with the canonical example $\chi_\text{can}(t) = \ee^{-t^2/2T_0^2}$ (\ref{example-harvesting}), we will impose $||\chi|| = ||\chi_\text{can}||$, so that the rescaled negativity can be written as
\begin{equation}\label{eq:N-tilde}
    \widetilde{\mathcal{N}}(\bm c) = \sqrt{\pi} \, \frac{|\trans{\bm{c}} \mf G \,\bm c| - \trans{\bm{c}} \mf W\, \bm c}{\trans{\bm c}\bm c},
\end{equation}
Throughout this section, we will maximize $\widetilde{\mathcal{N}}(\bm c)$, first when the detectors are spacelike separated, second when there is a small but non-negligible amount of signalling between them, and third when the detectors are causally connected, but the signalling between them is exactly zero.

We highlight once again that the optimization presented in this section takes place within the regime of second order perturbation theory, so that, in principle, $\lambda$ should be thought of as an infinitesimal parameter that ensures that higher order corrections are negligible. 

Overall, the improvements presented here should be understood relative to the canonical example~(\ref{example-harvesting}). Explicitly, the negativity increases we find can be understood as physical for setups that can be described within second order perturbation theory after the optimization is considered. In Section~\ref{sec:limit} we will discuss in more detail the validity of second order perturbation theory in entanglement harvesting in light of the results of this section.

%\vspace{-6mm}

\subsection{Spacelike separated regions}\label{sec:spacelike}

In this subsection, we will find the optimal switching function for entanglement harvesting with spacelike separated regions. While, in principle, the setup presented in Section~\ref{sec:iv} does not allow one to have genuine spacelike separation due to the tails of Hermite functions, we can approximate compact support arbitrarily well by tuning the scaling parameter $T$ in~\eqref{eq:h(t)}. Indeed, any compactly supported function $\chi\in L^2([-a,a])$ can be expanded as
\begin{equation}\label{eq:chiNtoInfty}
    \chi(t) = \lim_{N\to\infty} \smash{\sum_{n=0}^N} \langle h_{n,N}(t),\chi\rangle h_{n,N}\phantom{\Big|}
\end{equation}
where $h_{n,N}(t) = h_n (t,T_N)$ and the sequence $\{T_N\}$ behaves asymptotically as $T_N \sim a/\sqrt{2N+1}$. Thus, any choice of the form 
\begin{equation}\label{eq:TN}
    \smash{T_N = \frac{L/2}{\sqrt{2N+1}}f(N),} \phantom{\Big|}
\end{equation}
where $f(N)\to 1$, would ensure that for sufficiently large $N$, detectors separated by a distance $L$ with switching functions $\chi$ of the form of Eq.~\eqref{eq:chiNtoInfty} are effectively causally disconnected. Throughout this subsection we will use
\begin{equation}\label{eq:f(N)}
    f(N) = \frac{\sqrt{2N+1}}{2 + \sqrt{2N+1}},
\end{equation}
as this choice has provided us with a consistent scaling even for relatively small values of $N$. We can quantify the maximal tails of any basis function $h_{n,N}(t)$ outside of the support region $[-L/2,L/2]$ for a given $N$ by computing 
\begin{equation}
    I(N) = \max_{n\leq N}\int_{\Gamma}\dd t \, |h_{n,N}(t)|^2 ,
\end{equation}
where $\Gamma = \mathbb{R}\setminus[-L/2,\,L/2]$. The maximum is achieved at $n=N$, and we find that $I(N)$ decreases exponentially in $N$, with $I(0) \sim 10^{-6}$, decaying to $10^{-13}$ at $N=200$. Alternatively, we can estimate the signalling between the coupling regions by computing 
\begin{equation}
    \Delta_{NN}/{ H_{NN}} \equiv {\Delta}(\Lambda_{\tc{a},N},\Lambda_{\tc{b},N})/{ H(\Lambda_{\tc{a},N},\Lambda_{\tc{b},N})},
\end{equation}
representing the ratio between signalling and field correlations between the interaction regions defined by $h_{n,N}(t)$. The ratio decreases exponentially in $N$, with $\Delta_{NN}/H_{NN} \sim 10^{-5}$ for $N = 0$ and $10^{-45}$ for $N = 200$. These results allow us to safely consider that the corresponding interaction regions are spacelike separated for large enough values of $N$.% In Fig.~\ref{figure:tails2} we can see that both quantities become negligible for large $N$, with $\ \Delta_{NN}/ H_{NN}$ decreasing exponentially.
%\vspace{-0.3mm}

% \begin{figure}[h!]
% \centering
% \includegraphics[width=\linewidth]{Figures/dual_axis_plot.pdf}
% \caption{Compact support and causal communication estimators $I(N)$ and $\Delta_{NN}/H_{NN}$ as a function of $N$.}
% \label{figure:tails2}
% \end{figure} 
Let us now recall the decomposition of the Feynman propagator in Eq.~\eqref{eq:feynman-propagator}, which we can write in terms of the matrices $\mf G, \mf H$, and  $\mathsf{\Delta}$ as
\begin{equation}
    \mf G = \tfrac{1}{2}\mf H + \tfrac{\ii}{2}\mf \Delta.   
\end{equation}
With the choice of $T_N$ as in Eqs.~\eqref{eq:TN} and~\eqref{eq:f(N)}, we can now effectively assume $\mathsf{\Delta}$ to be zero, yielding $\mf G \approx \frac{1}{2} \mf H$, and the negativity in Eq.~\eqref{eq:N-tilde} can be approximated by
\begin{equation}
    \widetilde{\mathcal{N}}^+(\bm c) \equiv \sqrt{\pi} \, \frac{\frac{1}{2}|\trans{\bm{c}} \mf H \,\bm c| - \trans{\bm{c}} \mf W \,\bm c}{\trans{\bm c}\bm c}.
\end{equation}
This quantity was first defined in~\cite{ericksonNew} as the \textit{harvested negativity}, as it considers that all correlations stem from the state dependent part of the non-local Feynman propagator. When maximizing the negativity in the protocol of entanglement harvesting, this is also the term that must be optimized to maximize genuine entanglement extraction from the field\footnote{Notice that if one attempted to maximize $\widetilde{\mathcal{N}}(\bm c)$, one would also indirectly be maximizing the communication between the detectors, encoded in $\trans{\bm c} \mf \Delta \bm c$.}.
We can now simply maximize $\widetilde{\mathcal{N}}^+$, leading to the optimization problem
\begin{equation}\label{eq:N+}
    \widetilde{\mathcal{N}}_\text{max}^+ \equiv \max_{\bm c \in \mathbb{R}^N}\widetilde{\mathcal{N}}^+(\bm c) 
    =\max_{\trans{\bm c}\bm c=1 } \sqrt{\pi} \bigg( \frac{1}{2}|\trans{\bm{c}} \mf H \,\bm c| - \trans{\bm{c}} \mf W \,\bm c\bigg).
\end{equation}
Although the expression above resembles an operator norm, the absolute value in $|\trans{\bm c} \mf H \,\bm c|$ requires more care. 
Using $|z| = \max_{\theta}(\Re(z)\cos\theta +\Im(z)\sin\theta)$, we can write 
\begin{equation}\label{eq:M-theta}
 \smash{\widetilde{\mathcal{N}}_\text{max}^+= \max_{\theta} \max_{\trans{\bm c}\bm c=1 }\mf M^+(\theta),}\phantom{\big|}
\end{equation}
where 
\begin{equation}\label{eq:matrix-M}
    \mf M^+(\theta) =\tfrac{1}{2} \big( \Re(\mf H) \cos\theta + \Im(\mf H)\sin\theta\big) - \Re(\mf W),
\end{equation}
where we note that for real vectors $\trans{\bm c} \mf W \bm c = \trans{\bm c}\! \Re(\mf W) \bm c$. Thus, the problem reduces to finding the largest eigenvalue of the real symmetric matrix $\mf M^+(\theta)$ for each $\theta\in[0,2\pi)$. We will now study the behavior of the negativity and signalling within the subspaces
\begin{equation}
    V_N = \smash{\bigg\{ \chi\in L^2(\mathbb{R}) : \chi(t) = \sum_{n=0}^N c_n h_{n,N}(t)\bigg\}.}\phantom{\Big|}
\end{equation}
\vspace*{-5mm}
\begin{figure}[h!]
\centering\includegraphics[width=8.6cm]{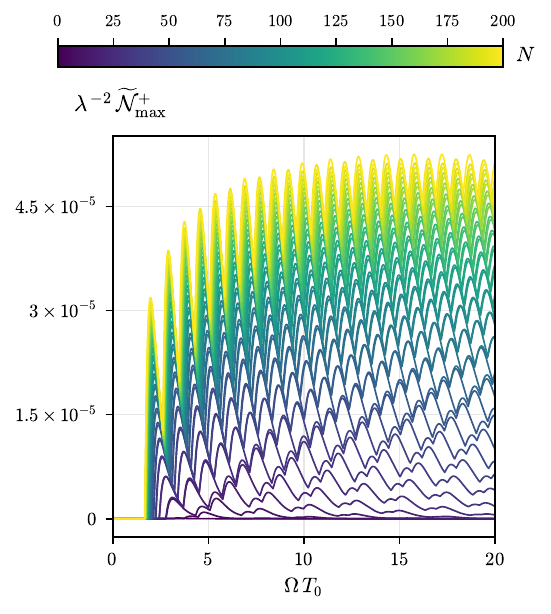}
    \vspace*{-10mm}
    \caption{Maximum negativity for detectors separated by a distance of $L = 5T_0$ within the subspace spanned by $\{h_{n,N}\}_N$ as a function of $\Omega$ for $N \in \{0, 200 \}$.}
\label{figure:neg-evolution-N}
\vspace{-1mm}
\end{figure}
\vspace{-3.4mm}

In Fig.~\ref{figure:neg-evolution-N}, we plot the maximum value of $\widetilde{\mathcal{N}}^+$ for switching functions within the subspaces $V_N$ as a function of $\Omega T_0$ when the detectors are separated by a distance of $|\bm L| = 5 T_0$. This choice allows us to directly compare our results with the canonical example discussed in Subsection~\ref{example-harvesting}. As the dimension $N$ of the subspaces $V_N$ increases, so does the height of the negativity peaks, as well as the value of $\Omega T_0$ for which the maximum negativity is achieved. We also see that for larger values of $N$ more  negativity peaks start to appear\footnote{For comparison, notice that the $N=0$ case corresponds to a rescaled version of Fig.~\ref{figure:example-gauss}, as $T_{N=0} = 5T_0/6$. The shorter interaction time also results in smaller negativity compared to the canonical example in Subsection~\ref{example-harvesting}.}. Notice that, even with the effective spacelike separation between the detectors implemented here, optimizing the negativity can yield values of the order of $10^{-5} \lambda^2$ for $N=200$, one order of magnitude larger than the canonical Gaussian example of Section~\ref{example-harvesting}. Also notice that, in accordance to the no-go theorems of~\cite{HarvestingQueNemLouko,nogo}, we see that no entanglement can be harvested at spacelike separation with gapless detectors ($\Omega = 0$).
%For each $N$, we see peaks for different Omegas, the peaks go up, no entanglement for Omega = 0 (cite~\cite{HarvestingQueNemLouko,nogo}), the value of Omega for which the actual maximum is acquired goes up for large N (for small N, the negativity is approximately decreasing with Omega, but then it starts to take longer in Omega to decrease).
\vspace*{-3mm}
\begin{figure}[h!]
\centering
\includegraphics[width=\linewidth]{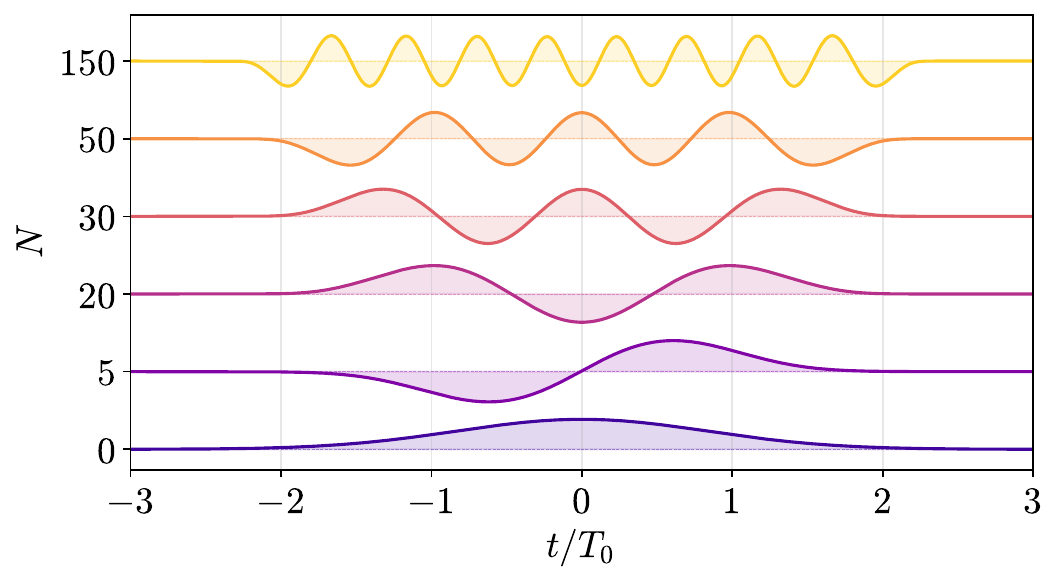}
\vspace*{-9mm}
\caption{Plot of the maximizing functions of $\widetilde{\mathcal{N}}$ within the subspace spanned by $\{h_{n,N}\}_N$ for detectors separated by a distance of $L = 5T_0$.}
\label{figure:oscillations}
\end{figure} 
\vspace*{-2mm}

In Fig.~\ref{figure:oscillations}, we show the shape of the optimal switching functions for different values of $N$. We note that, as $N$ increases, the functions become more oscillatory. The combination of oscillations in the switching function with the energy gap $\Omega$ determine the field frequencies that contribute the most to the detectors' final state. It is then no surprise that, in Fig~\ref{figure:neg-evolution-N}, higher values of $\Omega$ become relevant only for sufficiently large $N$, when the oscillations start taking place. Moreover, our results suggest that for large enough $N$, the optimal switching functions become superoscillatory, similar to the example given in~\cite{Reznik1}. %\shout{Mention the resonance with the peaks}
\begin{figure}[h!]
\centering
\includegraphics[width=\linewidth]{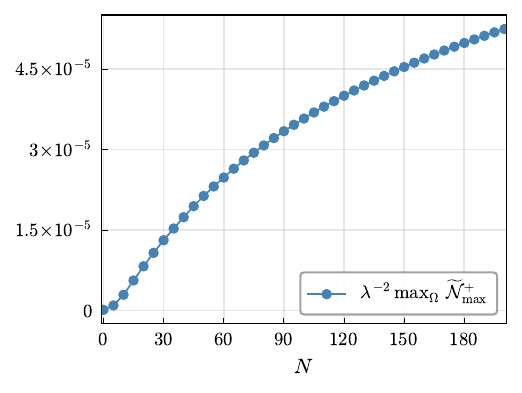}
\vspace*{-7mm}
\caption{Maximum negativity $\widetilde{\mathcal{N}}^+$  for detectors separated by a distance of $L = 5T_0$ on the subspace spanned by $\{h_{n,N}\}_N$ as a function of $N$.}
\label{figure:neg-evolution-N-real}
\end{figure} 

In Fig.~\ref{figure:neg-evolution-N-real}, we plot the maximal negativity $\widetilde{\mathcal N}_{\text{max}}^+(\Omega)$ also maximized over $\Omega$ for each $N$, corresponding to the largest value of Fig.~\ref{figure:neg-evolution-N} for each $N$. Notice that the plot does not asymptote with $N$. Indeed, in~\cite{Reznik1} the authors showed that a specific choice of superoscillating functions could result in negativities that increase as $\mathcal{N}^+ \sim \lambda^2 \sqrt{N_{\text{osc}}}$, with $N_{\text{osc}}$ being the number of oscillations of the function. Given that our basis expansion spans all compactly supported switching functions in the limit $N\to \infty$, we can conclude that Fig.~\ref{figure:neg-evolution-N-real} also diverges to infinity in this limit.
%\vspace{-6.5mm}

\begin{figure}[h!]
\centering
   \includegraphics[width=\linewidth]{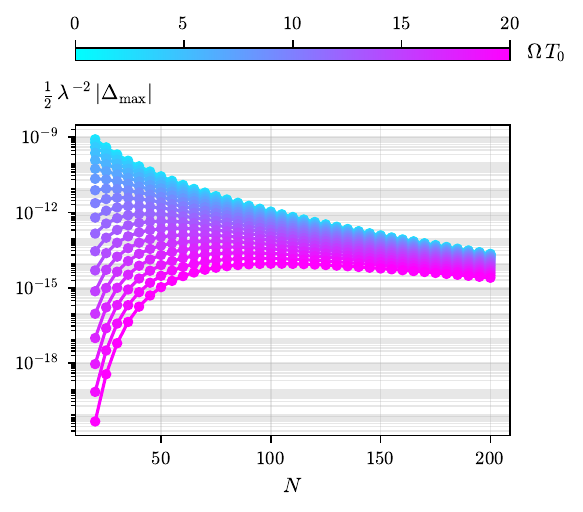}
   \vspace*{-10mm}
    \caption{Maximum eigenvalue of $\mf \Delta$ within the subspace spanned by $\{h_{n,N}\}_N$ for detectors separated by a distance of $L = 5T_0$ as a function of $N$ for different values of $\Omega$.}
\label{figure:delta_max}
\vspace*{-4mm}
\end{figure}

In Fig.~\ref{figure:delta_max} we plot the signalling estimator as a function of $N$ for different values of $\Omega$. Notice that, for all parameters used in this section, the signalling is negligible compared to the maximum negativities displayed in Fig.~\ref{figure:neg-evolution-N-real}. We see that the signalling increases for small values of $N$, which is explained by the fact that $T_N$ is not yet in the asymptotic behavior~\eqref{eq:TN}. However, at large enough $N$, we enter the regime in which the basis $\{h_{n,N}\}_N$ can be approximated to span compactly supported functions, where the signalling decreases exponentially with $N$ and reaches $10^{-14}$ at $N=200$. This results in a signalling-to-entanglement ratio of $\Theta \approx 10^{-8}$ for the maximal negativity values.

% \tmm{Describe behaviour of Delta as a function of $N$ and $\Omega$. Explain why for small $N$ we have $\Delta$ increasing because $T_N$ is not in the asymptotic behaviour yet. Mention that it is fully deacreasing with Omega, and mention that it is orders and orders (how many) of magnitude less than the negativities we are seeing in the other cases\footnote{\tmm{Notice that the $N=0$ case corresponds to a single Gaussian with standard deviation $T_{N=0} = 5T_0/6$, resulting in a shorter interaction time and, consequently, smaller negativity, compared to the canonical example in Subsection~\ref{example-harvesting}.}}.}

%In this subsection we studied the protocol of entanglement harvesting with two spacelike separated probes that couple to the field with arbitrary temporal profiles. In doing so, we found  

In general, it is not an easy task to find parameters and switching functions that allow spacelike separated probes to harvest entanglement from the 3+1 dimensional Minkowski vacuum (for examples see e.g.~\cite{NickEdu2014,HarvestingDelocalized,adrianQuenchedHarv2025,CongHorizons2020,patriciaAndI}). The results of this subsection showcased the advantage provided by the Hermite expansion method presented in Section~\ref{sec:iv}, providing numerous examples of switching profiles and overall parameters that can realize entanglement harvesting at spacelike separation. Moreover, we found setups where causally disconnected detectors harvest negativities one order of magnitude larger than the example with Gaussian switching functions, that has signalling-to-entanglement ratio of $\Theta \approx 5\%$.

%\shout{stopped here!}
%While it is hard conclude about such claims, as this would require a more thorough analysis of the asymptotic behavior of the negativity for large $N$, which is a highly non-trivial task (App.~\ref{app:asymp}), 

%Therefore, as we increase $N$, we recover a larger and larger maximum negativity, reaching even an order of magnitude over the standard Gaussian case at $N=200$. 
%Moreover, if we actually plot the largest negativity as a function of $N$, we can observe that it does not seem to asymptote for the values of $N$ considered (Fig.~\ref{figure:neg-evolution-N-real}).

%In this section we have shown that, by using the method presented in Sec.~\ref{sec:iv}, it it possible to harvest more entanglement using compactly supported switching functions that one would normally retrieve with a non compactly supported function, such as a Gaussian. In the next section we will relax the support region and use the same method to recover even more negativity using less compactly supported functions, while making sure the signalling stays below a certain threshold.

\subsection{Approximately causally disconnected regions}

Strictly speaking, it is not necessary for detectors to be fully causally disconnected to extract genuine vacuum entanglement from the field. Indeed, in the canonical example of Subsection~\ref{example-harvesting}, we had a non-negligible signalling-to-entanglement ratio (SER) of the order of $\Theta \approx 5\%$ at the peak of negativity. Even though in this case one cannot claim to extract entanglement from strictly spacelike separated regions, the SER is still sufficiently small for one to conclude that the detectors became entangled due to vacuum correlations.

%Looking back to the canonical case, we recall that the presence of some signalling could be beneficial, as it increases the amount of negativity that can be harvested. This begs the question: can we get even more orders of magnitude than in the previous section by allowing some signalling 

% In this section, we will find the optimal switching functions for entanglement harvesting from spacelike separated regions.

As a matter of fact, small signalling can increase the total entanglement acquired by the detectors, as it allows them to probe the field for longer periods of time. In this section, we take advantage of this fact to improve on the canonical example of Subsection~\ref{example-harvesting}, while keeping the SER $\Theta$ below $5\%$. To achieve this, we will apply the same method as in Subsection~\ref{sec:spacelike}, increasing the detectors' interaction times. Explicitly, we make $T_N \mapsto T_N+\delta T$ for small $\delta T$, rescaling the Hermite functions $h_{n}(t,T_N)$ and slightly increasing the signalling between the probes. To find the functions that maximize the genuine entanglement extracted from the field ($\widetilde{\mathcal{N}}^+$ in~\eqref{eq:N+}), we again find the maximal eigenvalues of $\mf M^+ (\theta)$ within the corresponding rescaled subspaces $V_N$ for each $\theta$. 

Explicitly, for each $N$, we pick multiple small values of $\delta T$ ensuring that the condition $\Theta\leq 5\%$ is satisfied. Importantly, for large values of $N$, the functions in the subspace $V_N$ become effectively compactly supported in $[-L/2,L/2]$, and the oscillating maximizing functions become very steep close to the boundary. For this reason, even very small rescaling parameters $\delta T$ yield non-trivial signalling, breaking the $\Theta \leq 5\%$ condition. Our results showed that $N\approx 50$ was an ideal balance between high negativity and low signalling with $\delta T/T_N \approx 2.4\%$.

\begin{figure}[h!]
    \centering
    \includegraphics[width=\linewidth]{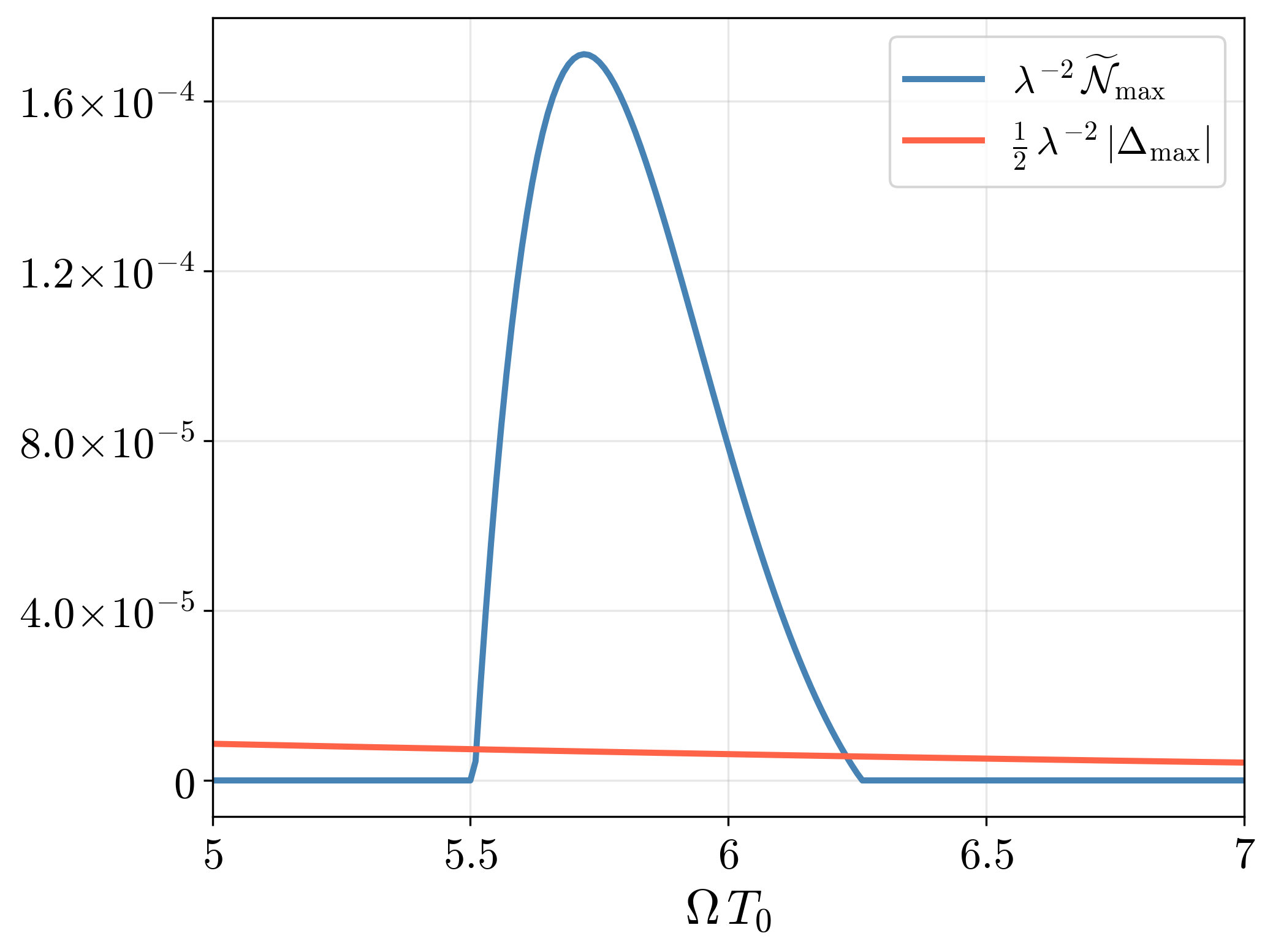}
    \vspace*{-9mm}
    \caption{Negativity plot for the optimal approximately disconnected switching function with $\Theta\leq 5\%$ at $N=50$ for detectors separated by a distance of $L = 5T_0$.}
\label{figure:best_peak_n50}
\end{figure}

Figure~\ref{figure:best_peak_n50} displays the equivalent negativity plot to the canonical example in Fig.~\ref{figure:example-gauss} for the optimal function when the detectors are separated by a distance $L = 5 T_0$ as a function of $\Omega T_0$ with $N=50$. In this example, the negativity peak reaches $10^{-4} \lambda^2$, while the signalling-to-entanglement at the peak is $4\%$, resulting in an improvement of two orders of magnitude. Also notice that the energy gap required in this case is larger than the one required in the Gaussian example, due to the oscillations in the switching functions (displayed in spacetime in Fig.~\ref{figure:best_traj_n50}).

%This non-zero signalling now comes from the fact that our optimal switching function is no longer compactly supported, as shown in Fig.~\ref{figure:best_traj_n50}, although it is still an oscillating function.
%In Figs.~\ref{figure:best_traj_n50}-\ref{figure:best_peak_n50}, we give the optimal switching function for $N=50$ and its corresponding negativity plot. The switching function has now tails outside of the support region, but it is still oscillating. 

\begin{figure}[h!]
    \centering
   \includegraphics[width=\linewidth]{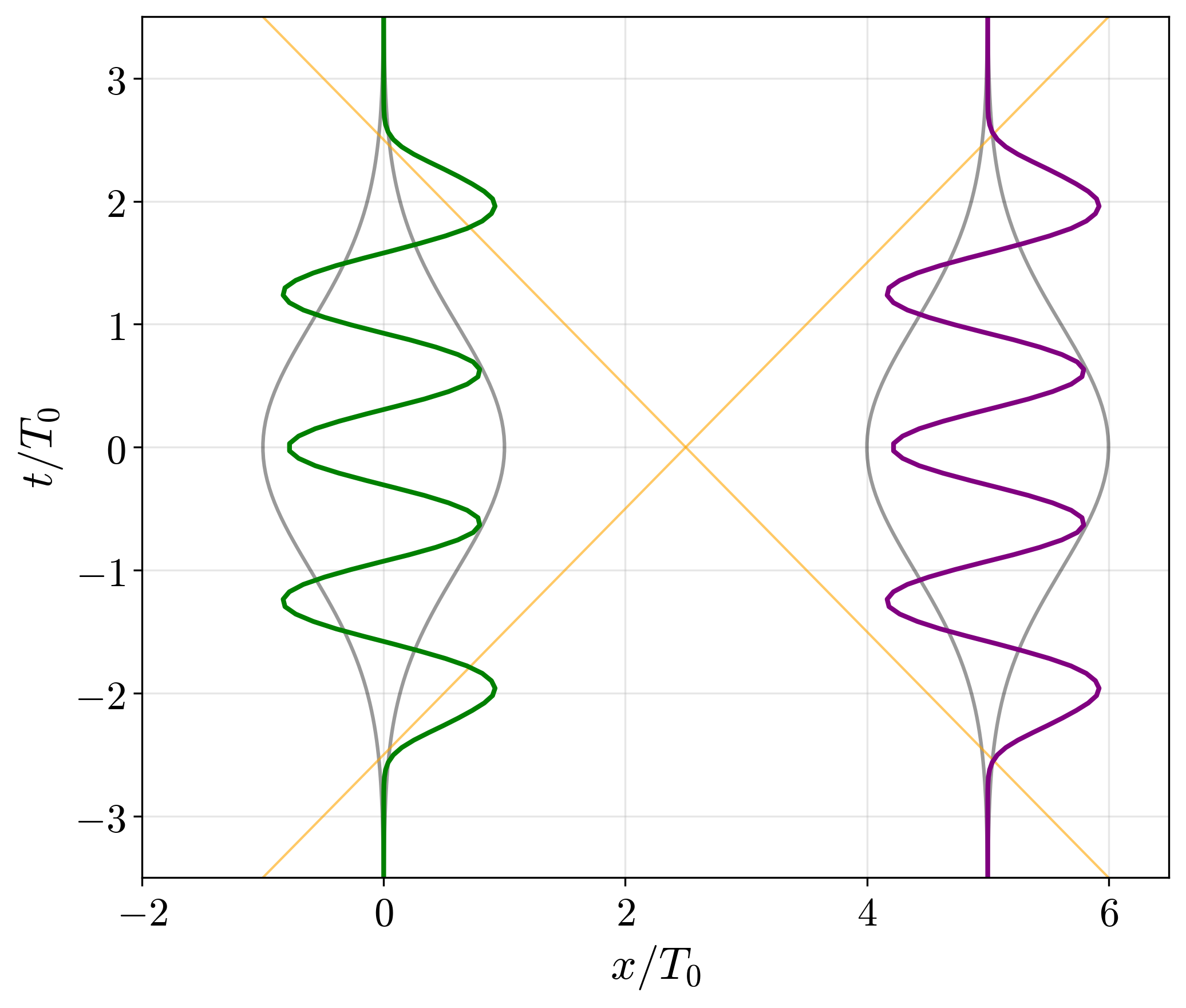}
   \vspace*{-7mm}
    \caption{Spacetime depiction of the optimal switching function at $N=50$ with the constraint that $\Theta\leq 5\%$ for detectors separated by a distance $L = 5T_0$. The profile of the switching functions are plotted in the horizontal axis around the spatial center of their interactions. The canonical Gaussian example is displayed in gray in the background for comparison and the orange lines depict the lightcone of the middle event between the interaction regions.}
\label{figure:best_traj_n50}
\end{figure}

In this subsection we applied our Hermite-expansion method to the case where the detectors can signal as much as in the canonical example of Subsection~\ref{example-harvesting}. By keeping the same level of signalling as in this example, we have shown that entanglement harvesting can be improved by two orders of magnitude. Moreover, when compared to the spacelike separated setup of Subsection~\ref{sec:spacelike}, we could also increase the entanglement between the probes by one order of magnitude by allowing small but non-negligible communication. 

\subsection{Non-communicating causally connected regions}

Even when the detectors are fully causally connected, it is possible that, for a specific choice of parameters, they do not communicate at leading order ($\Delta_{\tc{a}\tc{b}}^{\sct{++}} = 0$). This can be understood as a revival~\cite{RevivalJC1993,RevivalReview2004,RelQuantRevival2010} on the part\footnote{For an explicit example of this decomposition for gapless detectors, see~\cite{closedform2024}.} of the time evolution operator associated with the signalling between the qubits, when it effectively acts as an identity operator. In this case one cannot use the detectors to quantify entanglement between spacelike separated regions, as they couple to overlapping field degrees of freedom. Nevertheless, if the signalling-to-entanglement ratio is exactly zero, one can still claim that the detectors became entangled due to pre-existing entanglement in the field. In this section, we will use our Hermite-expansion method to find particular cases in which the probes, even if causally connected, do not communicate, but can harvest even more entanglement than the previously discussed cases.

%We will lift the constraint that the function need be compactly supported and we find particular cases in which the probes, even if causally connected, do not communicate but can harvest entanglement to the limit of perturbation theory without the need of superoscillations.

Notice that when considering $T$ sufficiently larger than $T_0$, the normalization condition based on the $L^2$ norm does not ensure that the peaks of the switching function are of the same order of magnitude as the examples we previously considered. More precisely, if $\chi$ is effectively supported in the region $[-\alpha L/2,\alpha L/2]$, the condition $||\chi|| = ||\chi_\text{can}||$ would result in \mbox{$\text{max}_t|\chi(t)| \sim \alpha^{-1}\text{max}_t|\chi_\text{can}(t)| = \alpha^{-1}$}. That is, increasing the interaction time would effectively decrease the total strength of the coupling $\lambda\chi(t)$. Therefore, an accurate comparison with the canonical case (assuming coupling constants of the same order of magnitude) would be achieved by the condition $||\chi|| = \alpha||\chi_\text{can}||$. Once imposed, it rescales the expressions for the negativity by a factor of $\alpha^2$ relative to the expression of Eq.~\eqref{eq:N+}. The same rescaling occurs for the signalling term $\Delta_\tc{ab}^{\sct{++}}$. In essence, when rescaling the time duration by $\alpha$, we will maximize the function
\begin{equation}
    \widetilde{\mathcal{N}}_\alpha(\bm c) \equiv \alpha^2\sqrt{\pi}  \, \frac{|\trans{\bm{c}} \mf G \,\bm c| - \trans{\bm{c}} \mf W \,\bm c}{\trans{\bm c}\bm c}
\end{equation}
over all $\bm c \in \mathbb{R}^N$ such that $\trans{\bm c}\mf \Delta\, \bm c = 0$ for a fixed $N$.  The constraint $\trans{\bm c}\mf \Delta \,\bm c = 0$ prevents this optimization from being recast as a simple eigenvalue problem, so we instead use an adapted quasi-Newton optimizer method to find the optimal switching function.

% We partially implement the constraint by splitting the signalling matrix onto its real and imaginary parts $\mathsf{\Delta} = \mathsf{\Delta}_\text{Re} + \ii \,\mathsf{\Delta}_\text{Im}$. 

% Ideally, one would find a specific $\bm c \equiv \bm c^+ + \bm c^-$ for which the positive and negative two terms cancel each other, i.e. $\trans{\bm c^+}\mathsf{\Delta}^+ \bm c^+ = \trans{\bm c^-}\mathsf{\Delta}^- \bm c^-$.

% Then, we compute the norm of each of these matrices, and we apply 

% To find $\bm c$ such that $\trans{\bm c}\mf \Delta_\text{Re} \bm c =\trans{\bm c}\mf \Delta_\text{Im} \bm c = 0$ we first split each of these real matrices onto their positive and negative parts $\mathsf{\Delta}_\text{Re} = \mathsf{\Delta}^+_\text{Re} - \mathsf{\Delta}^-_\text{Re}$, $\mathsf{\Delta}_\text{Im} = \mathsf{\Delta}^+_\text{Im} - \mathsf{\Delta}^-_\text{Im}$.  However, it becomes non-trivial to find a vector that works for both the real and imaginary part. Therefore, we perform this technique on one of the two, and then we use numerical methods

%making the signalling and negativity skyrocket. We will now give a prospect for finding a specific switching function where the full negativity is harvested, but the probes do no longer communicate, giving an effective zero signalling. In other words, we allow the signalling to occur in order to increase the negativity, and then we look for a switching function that effectively interrupts communication.

\begin{figure}[h!]
    \centering
   \includegraphics[width=\linewidth]{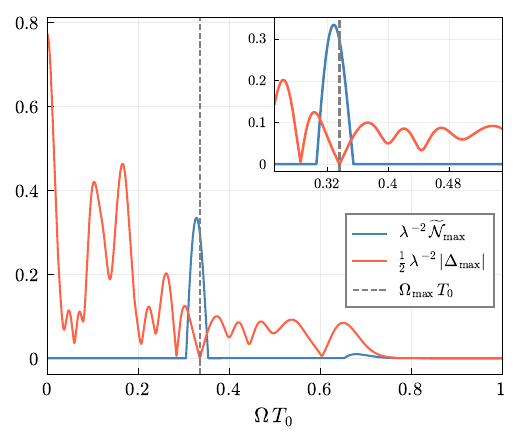}
   \vspace*{-10mm}
    \caption{Negativity and signalling contribution for the switching function found with our adapted quasi-Newton method for $N=25$, $T = 10 \,T_0$ and  $L = 5T_0$.}
\label{figure:best-neg-perturbation-theory}
\end{figure}

In Fig.~\ref{figure:best-neg-perturbation-theory}, we plot the negativity and signalling estimator as a function of $\Omega T_0$ for the optimal switching function with $N=25$ for an interaction lasting 10 times longer than the canonical example ($T=10T_0$). As expected, this setup has large signalling for most values of $\Omega$. However, at $\Omega_\text{max} T_0 \approx 0.336$, the signalling cancels exactly, corresponding to a point where $\Delta_\tc{ab}^{\sct{++}}$ changes in sign, while the negativity reaches $10^{-1}\lambda^2$.

Applying our Hermite-expansion method to causally connected regions, we were able to find specific switching functions and parameters such that the detectors cannot signal to each other to leading order in $\lambda$. The obtained negativity in this scenario was five orders of magnitude larger than the canonical example, and is expected to be even larger when the detectors probe the field for even larger times.

\section{The limit of second order perturbation theory}\label{sec:limit}

Entanglement harvesting is mostly treated within the regime of second order perturbation theory\footnote{We note that there are a few exceptions~\cite{Pozas-Kerstjens:2015,NickEdu2014,beyondPertDet2013,Farming,PetarNonpertEH,JoseNonPertMethod}}, using a formalism and expressions similar to the ones reviewed in Section~\ref{sec:iii}. The reasons for this approach are two-fold. On the one hand, this is the simplest theoretical framework to study entanglement harvesting, requiring the computation of a few integrals corresponding to the smeared propagators $G_\tc{f}$ and $W$. On the other hand, it is argued that the quantities defining the final state of the detectors (e.g. $\lambda^2 G_\tc{f}(\Lambda_\tc{a}^+,\Lambda_\tc{b}^+)$ and $\lambda^2 W(\Lambda_i^-,\Lambda_i^+)$), and thus, their entanglement, are small. Indeed, in the canonical Gaussian example, the relevant propagators are of the order of $ 10^{-5}\lambda^2$, so that even in the extreme case of $\lambda = 1$, leading order results yield a precise description of the protocol. In fact,~\cite{Pozas-Kerstjens:2015} shows that, when going to fourth order in the Gaussian canonical example, the next order terms are $\sim 10^{-10}\lambda^4$. 

Although the perturbative treatment seems to be enough, recent experimental proposals have been put forward to harvest entanglement from both analogue systems~\cite{cisco2023harvesting,oberthaler} and the electromagnetic vacuum~\cite{adamExperimentalEH2025,Mo_Settembrini2022VacuumCorrelationsNatCommun}. These approaches are all similar in the sense that the proposed experimental parameters yield negativity values that are, in principle, measurable with current methods. In the particular case of~\cite{cisco2023harvesting} and~\cite{oberthaler}, when using realistic experimental values for the coupling constant, one obtains negativities for the final state of the probes of the order of $\mathcal{N}\sim 10^{-4}$. This is possible because, in these cases, the coupling of probes happens to an analogue of the momentum of a scalar field $\hat{\pi}(\mf x)$, rather than to the field $\hat{\phi}(\mf x)$. For this type of interaction, the coupling constant is dimensionful, allowing experimental parameters of the setup to rescale the results. 

Another common feature of the experimental proposals~\cite{cisco2023harvesting,oberthaler,adamExperimentalEH2025,Mo_Settembrini2022VacuumCorrelationsNatCommun} is the fact that they all considered one short interaction pulse, corresponding to an even positive switching function. Although in some of the proposed experiments it is not possible to have the oscillating switching functions found in Section~\ref{sec:optimization}, the proposal of~\cite{oberthaler} is an exception. In this case, two localized polarons (playing the role of Unruh-DeWitt detectors) interact with the density fluctuations of a background Bose-Einstein condensate (analogous to the momentum of a real scalar field). The coupling between the probes and the condensate is controlled by an external magnetic field that can oscillate around a given value, allowing both positive and negative values for the switching function. 

One would hope that the results of Section~\ref{sec:optimization} would also carry on to momentum-coupled detectors in experiments that can implement oscillating switching functions. Indeed, our method was able to improve the negativities by two orders of magnitude compared to the canonical Gaussian protocol used in~\cite{oberthaler}, while keeping the same relative signalling. When considering non-communicating causally connected probes, the improvement was by five orders of magnitude instead. Similar improvements in the momentum-coupled case would take setups that predict negativities of the order of $10^{-4}$ beyond leading order perturbation theory. Moreover, one could further enhance the extracted entanglement by utilizing multiple detectors~\cite{sergiHarvesting}, improving the protocol even further.

Indeed, we can see that negativities of the order of $10^{-3}$ are already outside of the leading order regime by analyzing plots of $\mathcal{N}$ as a function of $\Omega$. For a given profile function, as $\Omega$ increases, the peak negativity \mbox{$\mathcal{N}(\Omega_\text{max}) = \lambda^2 G_\tc{f}(\Lambda_\tc{a}^+,\Lambda_\tc{b}^+)-\lambda^2 W(\Lambda_i^-,\Lambda_i^+)$} typically happens shortly after the point at which $\lambda^2 G_\tc{f}(\Lambda_\tc{a}^+,\Lambda_\tc{b}^+)$ and $\lambda^2 W(\Lambda_i^-,\Lambda_i^+)$ acquire the same value. The negativity, given by their difference, then tends to be one order of magnitude smaller than the propagators. That is, when the negativity is of the order of $10^{-3}$ or $10^{-2}$, the propagators will be of order $10^{-2}$ or $10^{-1}$, respectively. Thus, $(\lambda^2 W(\Lambda_i^-,\Lambda_i^+))^2$ and $(\lambda^2 G_\tc{f}(\Lambda_\tc{a}^+,\Lambda_\tc{b}^+))^2$ would have a similar magnitude to $\mathcal{N}$ and our perturbative treatment would not provide an accurate description to the experiment. These cases require considering higher order corrections, or, might even leave the perturbative regime altogether~\cite{JoseNonPertMethod}. However, our theoretical results still accurately describe the experimental setup~\cite{oberthaler} if one considers smaller values of the coupling constant, ensuring that the relevant corrections to the state are of order $10^{-3}$. This would increase the lifetime of probes and decrease experimental noise effects~\cite{MO_RbK_loss}, while keeping the description within second order perturbation theory and measurable negativity values.

The idea that local probes could be used to extract vacuum entanglement from a quantum field was first proposed in the context of atoms coupled to electromagnetism in the 90's~\cite{Valentini1991}. In the early studies of the protocol, it seemed evident that the effect was so small that the perturbative regime would be sufficient to describe it. However, recent experimental proposals would already almost reach the limit of second order perturbation theory. When combined with our results, the protocol could potentially be taken to the non-perturbative regime, begging a novel approach that could describe entanglement harvesting outside of perturbation theory.

\section{Conclusions}\label{sec:conclusions}

We optimized the protocol of entanglement harvesting and argued that, when applied to proposed experiments, our results can push the protocol beyond the standard perturbative regime. Our method uses a truncated Hermite expansion to efficiently compute smeared QFT propagators in closed-form, allowing us to consider arbitrary switching functions. With this method, the computation of negativity becomes a simple matrix product and optimizing genuine entanglement harvesting reduces to an eigenvalue problem. For any given setup, our method can pinpoint the optimal switching function that maximizes entanglement extraction, and we show that it can increase the harvested entanglement by several orders of magnitude in three different scenarios; when the probes are spacelike separated, approximately causally disconnected, and causally connected but non-communicating.

In the case of spacelike separated probes, we implement causal disconnection by rescaling the Hermite basis, and find oscillating switching functions that can harvest one order of magnitude more entanglement than the canonical example of a Gaussian switching function. We then adapt this method to the case where the probes are not fully spacelike separated by slightly increasing the temporal support of the Hermite basis. This gains two orders of magnitude compared to the Gaussian example, while ensuring less signaling. We estimate the communication between the probes through the signalling-to-entanglement ratio, a stricter quantifier of genuine entanglement harvesting than the estimator previously defined in~\cite{ericksonNew}. Finally, we consider probes that are causally connected and find specific parameters for which they are unable to communicate, finding a five order of magnitude increase of the extracted entanglement compared to the canonical example.

Up to this point, entanglement harvesting has only been treated perturbatively, resulting in an infinitesimal amount of extracted entanglement, and posing an experimental challenge in practice. However, extrapolating our conclusions to recent experimental proposals pushes the protocol beyond leading order perturbation theory, and perhaps even beyond the perturbative regime altogether. Our results make it clear that a non-perturbative treatment is the only way forward in the path to make entanglement harvesting useful in quantum information and quantum computing.

\acknowledgments

MMB and TRP are thankful to Markus K. Oberthaler for insightful discussions, to Alexander Flink for essential input on the numerical part of the work, and to Sebastian Holm\'en for reviewing a first draft of the manuscript. TRP is thankful for financial support from the Olle Engkvist Foundation (no.225-0062). Nordita is partially supported by Nordforsk. 

\phantom{test}

\bibliography{references}

@article{generalPD,
  title = {Localized nonrelativistic quantum systems in curved spacetimes: A general characterization of particle detector models},
  author = {Perche, T. Rick},
  journal = {Phys. Rev. D},
  volume = {106},
  issue = {2},
  pages = {025018},
  numpages = {20},
  year = {2022},
  month = {Jul},
  publisher = {American Physical Society},
  doi = {10.1103/PhysRevD.106.025018},
  url = {https://link.aps.org/doi/10.1103/PhysRevD.106.025018}
}

@article{Unruh1976,
  title = {Notes on black-hole evaporation},
  author = {Unruh, W. G.},
  journal = {Phys. Rev. D},
  volume = {14},
  issue = {4},
  pages = {870--892},
  numpages = {0},
  year = {1976},
  month = {Aug},
  publisher = {American Physical Society},
  doi = {10.1103/PhysRevD.14.870},
  url = {https://link.aps.org/doi/10.1103/PhysRevD.14.870}
}

@book{DeWitt,
	Address = {Cambridge, UK},
	Author = {B. DeWitt},
	Publisher = {Cambridge University Press},
	Title = {General Relativity; an Einstein Centenary Survey},
	Year = {1980}}

@article{us,
  title = {General relativistic quantum optics: Finite-size particle detector models in curved spacetimes},
  author = {Mart\'{\i}n-Mart\'{\i}nez, Eduardo and Perche, T. Rick and de S. L. Torres, Bruno},
  journal = {Phys. Rev. D},
  volume = {101},
  issue = {4},
  pages = {045017},
  numpages = {10},
  year = {2020},
  month = {Feb},
  publisher = {American Physical Society},
  doi = {10.1103/PhysRevD.101.045017},
  url = {https://link.aps.org/doi/10.1103/PhysRevD.101.045017}
}

@article{neutrinos,
  title = {Neutrino flavor oscillations without flavor states},
  author = {Torres, Bruno de S. L. and Rick Perche, T. and Landulfo, Andr\'e G. S. and Matsas, George E. A.},
  journal = {Phys. Rev. D},
  volume = {102},
  issue = {9},
  pages = {093003},
  numpages = {11},
  year = {2020},
  month = {Nov},
  publisher = {American Physical Society},
  doi = {10.1103/PhysRevD.102.093003},
  url = {https://link.aps.org/doi/10.1103/PhysRevD.102.093003}
}

@article{nogo,
  title = {General no-go theorem for entanglement extraction},
  author = {Simidzija, Petar and Jonsson, Robert H. and Mart\'{\i}n-Mart\'{\i}nez, Eduardo},
  journal = {Phys. Rev. D},
  volume = {97},
  issue = {12},
  pages = {125002},
  numpages = {16},
  year = {2018},
  month = {Jun},
  publisher = {American Physical Society},
  doi = {10.1103/PhysRevD.97.125002},
  url = {https://link.aps.org/doi/10.1103/PhysRevD.97.125002}
}

@misc{cisco2023harvesting,
      title={Vacuum entanglement probes for ultra-cold atom systems}, 
      author={Cisco Gooding and Allison Sachs and Robert B. Mann and Silke Weinfurtner},
      year={2023},
      eprint={2308.07892},
      archivePrefix={arXiv},
      primaryClass={quant-ph}
}

@article{phil,
  title = {Universal Quantum Computer from Relativistic Motion},
  author = {LeMaitre, Philip A. and Perche, T. Rick and Krumm, Marius and Briegel, Hans J.},
  journal = {Phys. Rev. Lett.},
  volume = {134},
  issue = {19},
  pages = {190601},
  numpages = {7},
  year = {2025},
  month = {May},
  publisher = {American Physical Society},
  doi = {10.1103/PhysRevLett.134.190601},
  url = {https://link.aps.org/doi/10.1103/PhysRevLett.134.190601}
}

@misc{adamExperimentalEH2025,
      title={Towards an experimental implementation of entanglement harvesting in superconducting circuits: effect of detector gap variation on entanglement harvesting}, 
      author={Adam Teixidó-Bonfill and Xi Dai and Adrian Lupascu and Eduardo Martín-Martínez},
      year={2025},
      eprint={2505.01516},
      archivePrefix={arXiv},
      primaryClass={quant-ph},
      url={https://arxiv.org/abs/2505.01516}, 
}

@article{Mo_Settembrini2022VacuumCorrelationsNatCommun,
  title   = {Detection of quantum-vacuum field correlations outside the light cone},
  author  = {Settembrini, Francesca Fabiana and Lindel, Frieder and Herter, Alexa Marina and Buhmann, Stefan Yoshi and Faist, J{\'e}r{\^o}me},
  journal = {Nature Communications},
  volume  = {13},
  number  = {1},
  pages   = {3383},
  year    = {2022},
  doi     = {10.1038/s41467-022-31081-1},
  url     = {https://doi.org/10.1038/s41467-022-31081-1}
}

@article{PRLHyugens2015,
  title = {{Violation of the Strong Huygen's Principle and Timelike Signals from the Early Universe}},
  author = {Blasco, Ana and Garay, Luis J. and Mart\'{\i}n-Benito, Mercedes and Mart\'{\i}n-Mart\'{\i}nez, Eduardo},
  journal = {Phys. Rev. Lett.},
  volume = {114},
  issue = {14},
  pages = {141103},
  numpages = {5},
  year = {2015},
  month = {Apr},
  publisher = {American Physical Society},
  doi = {10.1103/PhysRevLett.114.141103},
  url = {https://link.aps.org/doi/10.1103/PhysRevLett.114.141103}
}

@article{collectCalling,
  title = {Timelike information broadcasting in cosmology},
  author = {Blasco, Ana and Garay, Luis J. and Mart\'{\i}n-Benito, Mercedes and Mart\'{\i}n-Mart\'{\i}nez, Eduardo},
  journal = {Phys. Rev. D},
  volume = {93},
  issue = {2},
  pages = {024055},
  numpages = {17},
  year = {2016},
  month = {Jan},
  publisher = {American Physical Society},
  doi = {10.1103/PhysRevD.93.024055},
  url = {https://link.aps.org/doi/10.1103/PhysRevD.93.024055}
}

@misc{oberthaler,
      title={Bose polarons as relativistic {Unruh-DeWitt} detectors: Entanglement harvesting from {Bose-Einstein} condensates}, 
      author={T. Rick Perche and Francesco Gozzini and Markus K. Oberthaler},
      year={2026},
      eprint={2512.21381},
      archivePrefix={arXiv},
      primaryClass={quant-ph},
      url={https://arxiv.org/abs/2512.21381}, 
}

@article{JoseNonPertMethod,
  title = {Nonperturbative method for particle detectors with continuous interactions},
  author = {Polo-G\'omez, Jos\'e and Mart\'{\i}n-Mart\'{\i}nez, Eduardo},
  journal = {Phys. Rev. D},
  volume = {109},
  issue = {4},
  pages = {045014},
  numpages = {22},
  year = {2024},
  month = {Feb},
  publisher = {American Physical Society},
  doi = {10.1103/PhysRevD.109.045014},
  url = {https://link.aps.org/doi/10.1103/PhysRevD.109.045014}
}

@article{us2,
  title = {Broken covariance of particle detector models in relativistic quantum information},
  author = {Mart\'{\i}n-Mart\'{\i}nez, Eduardo and Perche, T. Rick and Torres, Bruno de S. L.},
  journal = {Phys. Rev. D},
  volume = {103},
  issue = {2},
  pages = {025007},
  numpages = {14},
  year = {2021},
  month = {Jan},
  publisher = {American Physical Society},
  doi = {10.1103/PhysRevD.103.025007},
  url = {https://link.aps.org/doi/10.1103/PhysRevD.103.025007}
}

@article{martin-martinez2015,
  title = {Causality issues of particle detector models in {QFT} and quantum optics},
  author = {Mart\'{\i}n-Mart\'{\i}nez, Eduardo},
  journal = {Phys. Rev. D},
  volume = {92},
  issue = {10},
  pages = {104019},
  numpages = {18},
  year = {2015},
  month = {Nov},
  publisher = {American Physical Society},
  doi = {10.1103/PhysRevD.92.104019},
  url = {https://link.aps.org/doi/10.1103/PhysRevD.92.104019}
}

@misc{QFTPD,
      title={{Particle Detectors from Localized Quantum Field Theories}}, 
      author={T. Rick Perche and Jos\'e Polo-G\'omez and Bruno de S. L. Torres and Eduardo Mart\'in-Mart\'inez},
      year={2023},
      eprint={2308.11698},
      archivePrefix={arXiv},
      primaryClass={quant-ph}
}

@article{Jonsson2,
  title = {Information Transmission Without Energy Exchange},
  author = {Jonsson, Robert H. and Mart\'{i}n-Mart\'{i}nez, Eduardo and Kempf, Achim},
  journal = {Phys. Rev. Lett.},
  volume = {114},
  issue = {11},
  pages = {110505},
  numpages = {5},
  year = {2015},
  month = {Mar},
  publisher = {American Physical Society},
  doi = {10.1103/PhysRevLett.114.110505},
  url = {https://link.aps.org/doi/10.1103/PhysRevLett.114.110505}
}

@article{eduardo,
  title = {Relativistic quantum optics: The relativistic invariance of the light-matter interaction models},
  author = {Mart\'{i}n-Mart\'{i}nez, Eduardo and Rodriguez-Lopez, Pablo},
  journal = {Phys. Rev. D},
  volume = {97},
  issue = {10},
  pages = {105026},
  numpages = {11},
  year = {2018},
  month = {May},
  publisher = {American Physical Society},
  doi = {10.1103/PhysRevD.97.105026},
  url = {https://link.aps.org/doi/10.1103/PhysRevD.97.105026}
}

@article{Nicho1,
  title = {{$\hat{\bm p}\cdot\hat{\bm{A}}$} vs {$\hat{\bm x}\cdot\hat{\bm{E}}$}: Gauge invariance in quantum optics and quantum field theory},
  author = {Funai, Nicholas and Louko, Jorma and Mart\'{\i}n-Mart\'{\i}nez, Eduardo},
  journal = {Phys. Rev. D},
  volume = {99},
  issue = {6},
  pages = {065014},
  numpages = {19},
  year = {2019},
  month = {Mar},
  publisher = {American Physical Society},
  doi = {10.1103/PhysRevD.99.065014},
  url = {https://link.aps.org/doi/10.1103/PhysRevD.99.065014}
}

@article{HarvestingSuperposed,
  title = {Entanglement amplification between superposed detectors in flat and curved spacetimes},
  author = {Foo, Joshua and Mann, Robert B. and Zych, Magdalena},
  journal = {Phys. Rev. D},
  volume = {103},
  issue = {6},
  pages = {065013},
  numpages = {19},
  year = {2021},
  month = {Mar},
  publisher = {American Physical Society},
  doi = {10.1103/PhysRevD.103.065013},
  url = {https://link.aps.org/doi/10.1103/PhysRevD.103.065013}
}

@article{HarvestingAccelerationRobb,
  title = {Does acceleration assist entanglement harvesting?},
  author = {Liu, Zhihong and Zhang, Jialin and Mann, Robert B. and Yu, Hongwei},
  journal = {Phys. Rev. D},
  volume = {105},
  issue = {8},
  pages = {085012},
  numpages = {9},
  year = {2022},
  month = {Apr},
  publisher = {American Physical Society},
  doi = {10.1103/PhysRevD.105.085012},
  url = {https://link.aps.org/doi/10.1103/PhysRevD.105.085012}
}

@article{HarvestingDelocalized,
  title = {Entanglement harvesting with coherently delocalized matter},
  author = {Stritzelberger, Nadine and Henderson, Laura J. and Baccetti, Valentina and Menicucci, Nicolas C. and Kempf, Achim},
  journal = {Phys. Rev. D},
  volume = {103},
  issue = {1},
  pages = {016007},
  numpages = {14},
  year = {2021},
  month = {Jan},
  publisher = {American Physical Society},
  doi = {10.1103/PhysRevD.103.016007},
  url = {https://link.aps.org/doi/10.1103/PhysRevD.103.016007}
}

@article{HarvestingQueNemLouko,
  title = {Degenerate detectors are unable to harvest spacelike entanglement},
  author = {Pozas-Kerstjens, Alejandro and Louko, Jorma and Mart\'{\i}n-Mart\'{\i}nez, Eduardo},
  journal = {Phys. Rev. D},
  volume = {95},
  issue = {10},
  pages = {105009},
  numpages = {11},
  year = {2017},
  month = {May},
  publisher = {American Physical Society},
  doi = {10.1103/PhysRevD.95.105009},
  url = {https://link.aps.org/doi/10.1103/PhysRevD.95.105009}
}

@article{VidalNegativity,
  title = {Computable measure of entanglement},
  author = {Vidal, G. and Werner, R. F.},
  journal = {Phys. Rev. A},
  volume = {65},
  issue = {3},
  pages = {032314},
  numpages = {11},
  year = {2002},
  month = {Feb},
  publisher = {American Physical Society},
  doi = {10.1103/PhysRevA.65.032314},
  url = {https://link.aps.org/doi/10.1103/PhysRevA.65.032314}
}

@article{ericksonNew,
  title = {When entanglement harvesting is not really harvesting},
  author = {Tjoa, Erickson and Mart\'{\i}n-Mart\'{\i}nez, Eduardo},
  journal = {Phys. Rev. D},
  volume = {104},
  issue = {12},
  pages = {125005},
  numpages = {21},
  year = {2021},
  month = {Dec},
  publisher = {American Physical Society},
  doi = {10.1103/PhysRevD.104.125005},
  url = {https://link.aps.org/doi/10.1103/PhysRevD.104.125005}
}

@article{boris,
  title = {Harvesting entanglement from the gravitational vacuum},
  author = {Perche, T. Rick and Ragula, Boris and Mart\'{\i}n-Mart\'{\i}nez, Eduardo},
  journal = {Phys. Rev. D},
  volume = {108},
  issue = {8},
  pages = {085025},
  numpages = {58},
  year = {2023},
  month = {Oct},
  publisher = {American Physical Society},
  doi = {10.1103/PhysRevD.108.085025},
  url = {https://link.aps.org/doi/10.1103/PhysRevD.108.085025}
}

@inbook{Fredenhagen2015utr,
    author = "Fredenhagen, Klaus",
    editor = "Brunetti, Romeo and Dappiaggi, Claudio and Fredenhagen, Klaus and Yngvason, Jakob",
    title = "{An Introduction to Algebraic Quantum Field Theory}",
    booktitle = "{Advances in algebraic quantum field theory}",
    doi = "10.1007/978-3-319-21353-8_1",
    pages = "1--30",
    year = "2015"
}

@Article{Cong2019,
author={Cong, Wan
and Tjoa, Erickson
and Mann, Robert B.},
title={Entanglement harvesting with moving mirrors},
journal={J. High Energy Phys.},
year={2019},
month={Jun},
day={07},
volume={2019},
number={6},
pages={21},
issn={1029-8479},
doi={10.1007/JHEP06(2019)021},
url={https://doi.org/10.1007/JHEP06(2019)021}
}

@article{RalphOlson2,
  title = {Extraction of timelike entanglement from the quantum vacuum},
  author = {Olson, S. Jay and Ralph, Timothy C.},
  journal = {Phys. Rev. A},
  volume = {85},
  issue = {1},
  pages = {012306},
  numpages = {7},
  year = {2012},
  month = {Jan},
  publisher = {American Physical Society},
  doi = {10.1103/PhysRevA.85.012306},
  url = {https://link.aps.org/doi/10.1103/PhysRevA.85.012306}
}

@article{hectorMass,
  title = {Entanglement harvesting: Detector gap and field mass optimization},
  author = {Maeso-Garc\'{\i}a, H\'ector and Perche, T. Rick and Mart\'{\i}n-Mart\'{\i}nez, Eduardo},
  journal = {Phys. Rev. D},
  volume = {106},
  issue = {4},
  pages = {045014},
  numpages = {15},
  year = {2022},
  month = {Aug},
  publisher = {American Physical Society},
  doi = {10.1103/PhysRevD.106.045014},
  url = {https://link.aps.org/doi/10.1103/PhysRevD.106.045014}
}

@article{Pozas2016,
  title = {Entanglement harvesting from the electromagnetic vacuum with hydrogenlike atoms},
  author = {Pozas-Kerstjens, Alejandro and Mart\'{i}n-Mart\'{i}nez, Eduardo},
  journal = {Phys. Rev. D},
  volume = {94},
  issue = {6},
  pages = {064074},
  numpages = {27},
  year = {2016},
  month = {Sep},
  publisher = {American Physical Society},
  doi = {10.1103/PhysRevD.94.064074},
  url = {https://link.aps.org/doi/10.1103/PhysRevD.94.064074}
}

@article{NickEdu2014,
	doi = {10.1088/0264-9381/31/21/214001},
	url = {https://doi.org/10.1088/0264-9381/31/21/214001},
	year = 2014,
	month = {oct},
	publisher = {{IOP} Publishing},
	volume = {31},
	number = {21},
	pages = {214001},
	author = {Eduardo Mart{\'{\i}}n-Mart{\'{\i}}nez and Nicolas C Menicucci},
	title = {Entanglement in curved spacetimes and cosmology},
	journal = {Class. Quantum Gravity}
}

@article{Pozas-Kerstjens:2015,
	Author = {Pozas-Kerstjens, Alejandro and Mart\'{i}n-Mart\'{i}nez, Eduardo},
	Doi = {10.1103/PhysRevD.92.064042},
	Issue = {6},
	Journal = {Phys. Rev. D},
	Month = {Sep},
	Numpages = {18},
	Pages = {064042},
	Publisher = {American Physical Society},
	Title = {Harvesting correlations from the quantum vacuum},
	Url = {http://link.aps.org/doi/10.1103/PhysRevD.92.064042},
	Volume = {92},
	Year = {2015}}

@article{reznik2,
  title = {Long-range entanglement in the {D}irac vacuum},
  author = {Silman, J. and Reznik, B.},
  journal = {Phys. Rev. A},
  volume = {75},
  issue = {5},
  pages = {052307},
  numpages = {5},
  year = {2007},
  month = {May},
  publisher = {American Physical Society},
  doi = {10.1103/PhysRevA.75.052307},
  url = {https://link.aps.org/doi/10.1103/PhysRevA.75.052307}
}

@article{HarvestingBHLaura,
	doi = {10.1088/1361-6382/aae27e},
	url = {https://doi.org/10.1088%2F1361-6382%2Faae27e},
	year = 2018,
	month = {oct},
	publisher = {{IOP} Publishing},
	volume = {35},
	number = {21},
	pages = {21LT02},
	author = {Laura J Henderson and Robie A Hennigar and Robert B Mann and Alexander R H Smith and Jialin Zhang},
	title = {Harvesting entanglement from the black hole vacuum},
	journal = {Class. Quantum Gravity}
}

@article{Reznik1,
	Author = {Benni Reznik and Alex Retzker and Jonathan Silman},
	Eid = {042104},
	Journal = {Phys. Rev. A},
	Number = {4},
	Numpages = {4},
	Pages = {042104},
	Publisher = {APS},
	Title = {{Violating Bell's inequalities in vacuum}},
	Url = {http://link.aps.org/abstract/PRA/v71/e042104},
	Volume = {71},
	Year = {2005}}

@article{Ng1,
  title = {New techniques for entanglement harvesting in flat and curved spacetimes},
  author = {Ng, Keith K. and Mann, Robert B. and Mart\'{\i}n-Mart\'{\i}nez, Eduardo},
  journal = {Phys. Rev. D},
  volume = {97},
  issue = {12},
  pages = {125011},
  numpages = {8},
  year = {2018},
  month = {Jun},
  publisher = {American Physical Society},
  doi = {10.1103/PhysRevD.97.125011},
  url = {https://link.aps.org/doi/10.1103/PhysRevD.97.125011}
}

@article{Valentini1991,
	Author = {Antony Valentini},
	Doi = {http://dx.doi.org/10.1016/0375-9601(91)90952-5},
	Issn = {0375-9601},
	Journal = {Phys. Lett. A},
	Number = {6-7},
	Pages = {321 - 325},
	Title = {Non-local correlations in quantum electrodynamics},
	Url = {http://www.sciencedirect.com/science/article/pii/0375960191909525},
	Volume = {153},
	Year = {1991}}

@article{Farming,
	Author = {Mart\'{i}n-Mart\'{i}nez, Eduardo and Brown, Eric G. and Donnelly, William and Kempf, Achim},
	Doi = {10.1103/PhysRevA.88.052310},
	Issue = {5},
	Journal = {Phys. Rev. A},
	Month = {Nov},
	Numpages = {15},
	Pages = {052310},
	Publisher = {American Physical Society},
	Title = {Sustainable entanglement production from a quantum field},
	Url = {http://link.aps.org/doi/10.1103/PhysRevA.88.052310},
	Volume = {88},
	Year = {2013}}

@article{witten,
  title = {APS Medal for Exceptional Achievement in Research: Invited article on entanglement properties of quantum field theory},
  author = {Witten, Edward},
  journal = {Rev. Mod. Phys.},
  volume = {90},
  issue = {4},
  pages = {045003},
  numpages = {38},
  year = {2018},
  month = {Oct},
  publisher = {American Physical Society},
  doi = {10.1103/RevModPhys.90.045003},
  url = {https://link.aps.org/doi/10.1103/RevModPhys.90.045003}
}

@article{Unruh-Wald,
	Author = {Unruh, William G. and Wald, Robert M.},
	Doi = {10.1103/PhysRevD.29.1047},
	Issue = {6},
	Journal = {Phys. Rev. D},
	Month = {Mar},
	Pages = {1047--1056},
	Publisher = {American Physical Society},
	Title = {{What happens when an accelerating observer detects a Rindler particle}},
	Volume = {29},
	Year = {1984}
}

@article{richard,
  title = {Quantum delocalization, gauge, and quantum optics: Light-matter interaction in relativistic quantum information},
  author = {Lopp, Richard and Mart\'{i}n-Mart\'{i}nez, Eduardo},
  journal = {Phys. Rev. A},
  volume = {103},
  issue = {1},
  pages = {013703},
  numpages = {20},
  year = {2021},
  month = {Jan},
  doi = {10.1103/PhysRevA.103.013703},
  url = {https://link.aps.org/doi/10.1103/PhysRevA.103.013703}
}

@article{HottaDistance,
  title = {Quantum energy teleportation without a limit of distance},
  author = {Hotta, Masahiro and Matsumoto, Jiro and Yusa, Go},
  journal = {Phys. Rev. A},
  volume = {89},
  issue = {1},
  pages = {012311},
  numpages = {6},
  year = {2014},
  month = {Jan},
  publisher = {American Physical Society},
  doi = {10.1103/PhysRevA.89.012311},
  url = {https://link.aps.org/doi/10.1103/PhysRevA.89.012311}
}

@article{nichoTeleport,
  title = {Engineering negative stress-energy densities with quantum energy teleportation},
  author = {Funai, Nicholas and Mart\'{\i}n-Mart\'{\i}nez, Eduardo},
  journal = {Phys. Rev. D},
  volume = {96},
  issue = {2},
  pages = {025014},
  numpages = {16},
  year = {2017},
  month = {Jul},
  publisher = {American Physical Society},
  doi = {10.1103/PhysRevD.96.025014},
  url = {https://link.aps.org/doi/10.1103/PhysRevD.96.025014}
}

@article{teleportation,
  title = {Quantum measurement information as a key to energy extraction from local vacuums},
  author = {Hotta, Masahiro},
  journal = {Phys. Rev. D},
  volume = {78},
  issue = {4},
  pages = {045006},
  numpages = {9},
  year = {2008},
  month = {Aug},
  publisher = {American Physical Society},
  doi = {10.1103/PhysRevD.78.045006},
  url = {https://link.aps.org/doi/10.1103/PhysRevD.78.045006}
}

@misc{Hotta2011,
  title={{Q}uantum {E}nergy {T}eleportation: {A}n {I}ntroductory {R}eview}, 
  author={Masahiro Hotta},
  year={2011},
  eprint={1101.3954},
  archivePrefix={arXiv},
  primaryClass={quant-ph}
}

@article{carol,
  title = {Harvesting entanglement from complex scalar and fermionic fields with linearly coupled particle detectors},
  author = {Perche, T. Rick and Lima, Caroline and Mart\'{\i}n-Mart\'{\i}nez, Eduardo},
  journal = {Phys. Rev. D},
  volume = {105},
  issue = {6},
  pages = {065016},
  numpages = {24},
  year = {2022},
  month = {Mar},
  publisher = {American Physical Society},
  doi = {10.1103/PhysRevD.105.065016},
  url = {https://link.aps.org/doi/10.1103/PhysRevD.105.065016}
}

@article{pitelli,
  title = {Angular momentum based graviton detector},
  author = {Pitelli, J. P. M. and Perche, T. Rick},
  journal = {Phys. Rev. D},
  volume = {104},
  issue = {6},
  pages = {065016},
  numpages = {9},
  year = {2021},
  month = {Sep},
  publisher = {American Physical Society},
  doi = {10.1103/PhysRevD.104.065016},
  url = {https://link.aps.org/doi/10.1103/PhysRevD.104.065016}
}

@article{B,
  title = {Spin Entanglement Witness for Quantum Gravity},
  author = {Bose, Sougato and Mazumdar, Anupam and Morley, Gavin W. and Ulbricht, Hendrik and Toro\ifmmode \check{s}\else \v{s}\fi{}, Marko and Paternostro, Mauro and Geraci, Andrew A. and Barker, Peter F. and Kim, M. S. and Milburn, Gerard},
  journal = {Phys. Rev. Lett.},
  volume = {119},
  issue = {24},
  pages = {240401},
  numpages = {6},
  year = {2017},
  month = {Dec},
  publisher = {American Physical Society},
  doi = {10.1103/PhysRevLett.119.240401},
  url = {https://link.aps.org/doi/10.1103/PhysRevLett.119.240401}
}

@article{areaLaw1993,
  title = {Entropy and area},
  author = {Srednicki, Mark},
  journal = {Phys. Rev. Lett.},
  volume = {71},
  issue = {5},
  pages = {666--669},
  numpages = {0},
  year = {1993},
  month = {Aug},
  publisher = {American Physical Society},
  doi = {10.1103/PhysRevLett.71.666},
  url = {https://link.aps.org/doi/10.1103/PhysRevLett.71.666}
}

@article{areaLawReview2010,
  title = {Colloquium: Area laws for the entanglement entropy},
  author = {Eisert, J. and Cramer, M. and Plenio, M. B.},
  journal = {Rev. Mod. Phys.},
  volume = {82},
  issue = {1},
  pages = {277--306},
  numpages = {0},
  year = {2010},
  month = {Feb},
  publisher = {American Physical Society},
  doi = {10.1103/RevModPhys.82.277},
  url = {https://link.aps.org/doi/10.1103/RevModPhys.82.277}
}

@Article{kelly,
author={de S. L. Torres, Bruno
and Wurtz, Kelly
and Polo-G{\'o}mez, Jos{\'e}
and Mart{\'i}n-Mart{\'i}nez, Eduardo},
title={Entanglement structure of quantum fields through local probes},
journal={J. High Energy Phys.},
year={2023},
month={May},
day={08},
volume={2023},
number={5},
pages={58},
issn={1029-8479},
doi={10.1007/JHEP05(2023)058},
url={https://doi.org/10.1007/JHEP05(2023)058}
}

@article{quantClass,
  title = {Role of quantum degrees of freedom of relativistic fields in quantum information protocols},
  author = {Perche, T. Rick and Mart\'{\i}n-Mart\'{\i}nez, Eduardo},
  journal = {Phys. Rev. A},
  volume = {107},
  issue = {4},
  pages = {042612},
  numpages = {20},
  year = {2023},
  month = {Apr},
  publisher = {American Physical Society},
  doi = {10.1103/PhysRevA.107.042612},
  url = {https://link.aps.org/doi/10.1103/PhysRevA.107.042612}
}

@article{EricksonZero,
  title = {Vacuum entanglement harvesting with a zero mode},
  author = {Tjoa, Erickson and Mart\'{\i}n-Mart\'{\i}nez, Eduardo},
  journal = {Phys. Rev. D},
  volume = {101},
  issue = {12},
  pages = {125020},
  numpages = {10},
  year = {2020},
  month = {Jun},
  publisher = {American Physical Society},
  doi = {10.1103/PhysRevD.101.125020},
  url = {https://link.aps.org/doi/10.1103/PhysRevD.101.125020}
}

@article{max,
doi = {10.1088/1361-6382/ac1b08},
url = {https://dx.doi.org/10.1088/1361-6382/ac1b08},
year = {2021},
month = {sep},
publisher = {IOP Publishing},
volume = {38},
number = {19},
pages = {195029},
author = {Maximilian H Ruep},
title = {Weakly coupled local particle detectors cannot harvest entanglement},
journal = {Class. Quantum Gravity}
}

@misc{sergiHarvesting,
      title={Bipartite entanglement harvesting with multiple detectors}, 
      author={Santeri Salomaa and Esko Keski-Vakkuri and Sergi Nadal-Gisbert},
      year={2026},
      eprint={2604.13869},
      archivePrefix={arXiv},
      primaryClass={quant-ph},
      url={https://arxiv.org/abs/2604.13869}, 
}

@misc{EHEapp,
    title = {Entanglement {H}arvesting {E}xplorer},
    author = {Morote-Balboa, Marcos and Perche, T. Rick},
    year = {2026},
    url = {https://entanglement-harvesting-explorer.streamlit.app/},
    note = {Online tool for computing negativity and downloading numerical data.}
}

@article{FullyRelativisticEH,
  title = {Fully relativistic entanglement harvesting},
  author = {Perche, T. Rick and Polo-G\'omez, Jos\'e and Torres, Bruno de S. L. and Mart\'{\i}n-Mart\'{\i}nez, Eduardo},
  journal = {Phys. Rev. D},
  volume = {109},
  issue = {4},
  pages = {045018},
  numpages = {17},
  year = {2024},
  month = {Feb},
  publisher = {American Physical Society},
  doi = {10.1103/PhysRevD.109.045018},
  url = {https://link.aps.org/doi/10.1103/PhysRevD.109.045018}
}

@article{closedform2024,
  title = {Closed-form expressions for smeared bidistributions of a massless scalar field: Nonperturbative and asymptotic results in relativistic quantum information},
  author = {Perche, T. Rick},
  journal = {Phys. Rev. D},
  volume = {110},
  issue = {2},
  pages = {025013},
  numpages = {21},
  year = {2024},
  month = {Jul},
  publisher = {American Physical Society},
  doi = {10.1103/PhysRevD.110.025013},
  url = {https://link.aps.org/doi/10.1103/PhysRevD.110.025013}
}

@article{Horodecki_1996,
   title={Separability of mixed states: necessary and sufficient conditions},
   volume={223},
   ISSN={0375-9601},
   url={http://dx.doi.org/10.1016/S0375-9601(96)00706-2},
   DOI={10.1016/s0375-9601(96)00706-2},
   number={1–2},
   journal={Physics Letters A},
   publisher={Elsevier BV},
   author={Horodecki, Michał and Horodecki, Paweł and Horodecki, Ryszard},
   year={1996},
   month=nov, pages={1–8} }

@article{Peres_1996,
   title={Separability Criterion for Density Matrices},
   volume={77},
   ISSN={1079-7114},
   url={http://dx.doi.org/10.1103/PhysRevLett.77.1413},
   DOI={10.1103/physrevlett.77.1413},
   number={8},
   journal={Physical Review Letters},
   publisher={American Physical Society (APS)},
   author={Peres, Asher},
   year={1996},
   month=aug, pages={1413–1415} }

@article{RelQuantRevival2010,
  title = {Relativistic Quantum Revivals},
  author = {Strange, P.},
  journal = {Phys. Rev. Lett.},
  volume = {104},
  issue = {12},
  pages = {120403},
  numpages = {4},
  year = {2010},
  month = {Mar},
  publisher = {American Physical Society},
  doi = {10.1103/PhysRevLett.104.120403},
  url = {https://link.aps.org/doi/10.1103/PhysRevLett.104.120403}
}

@article{RevivalReview2004,
title = {Quantum wave packet revivals},
journal = {Physics Reports},
volume = {392},
number = {1},
pages = {1-119},
year = {2004},
issn = {0370-1573},
doi = {https://doi.org/10.1016/j.physrep.2003.11.002},
url = {https://www.sciencedirect.com/science/article/pii/S0370157303004381},
author = {R.W. Robinett}
}

@inbook{RevivalJC1993, place={Cambridge}, title={The Jaynes–Cummings Revival}, booktitle={Physics and Probability: Essays in Honor of Edwin T. Jaynes}, publisher={Cambridge University Press}, author={Shore, B.W. and Knight, P.L.}, editor={Grandy, Jr, W. T. and Milonni, P. W.Editors}, year={1993}, pages={15–32}}

@misc{borisTeleportRev,
      title={A review of applications of Quantum Energy Teleportation: from experimental tests to thermodynamics and spacetime engineering}, 
      author={Boris Ragula and Eduardo Martín-Martínez},
      year={2025},
      eprint={2505.04689},
      archivePrefix={arXiv},
      primaryClass={quant-ph},
      url={https://arxiv.org/abs/2505.04689}, 
}

@article{patriciaAndI,
  title = {Multimode nature of spacetime entanglement in QFT},
  author = {Agullo, Ivan and Bonga, B\'eatrice and Mart\'{\i}n-Mart\'{\i}nez, Eduardo and Nadal-Gisbert, Sergi and Perche, T. Rick and Polo-G\'omez, Jos\'e and Ribes-Metidieri, Patricia and Torres, Bruno de S. L.},
  journal = {Phys. Rev. D},
  volume = {111},
  issue = {8},
  pages = {085013},
  numpages = {23},
  year = {2025},
  month = {Apr},
  publisher = {American Physical Society},
  doi = {10.1103/PhysRevD.111.085013},
  url = {https://link.aps.org/doi/10.1103/PhysRevD.111.085013}
}

@Article{CongHorizons2020,
author={Cong, Wan
and Qian, Chen
and Good, Michael R.R.
and Mann, Robert B.},
title={Effects of horizons on entanglement harvesting},
journal={J. High Energy Phys.},
year={2020},
month={Oct},
day={12},
volume={2020},
number={10},
pages={67},
issn={1029-8479},
doi={10.1007/JHEP10(2020)067},
url={https://doi.org/10.1007/JHEP10(2020)067}
}

@article{adrianQuenchedHarv2025,
  title = {Quenched entanglement harvesting},
  author = {Lopez-Raven, Adrian and Mann, Robert B. and Louko, Jorma},
  journal = {Phys. Rev. D},
  volume = {112},
  issue = {8},
  pages = {085001},
  numpages = {10},
  year = {2025},
  month = {Oct},
  publisher = {American Physical Society},
  doi = {10.1103/lyhy-ftxz},
  url = {https://link.aps.org/doi/10.1103/lyhy-ftxz}
}

@article{plenio,
  title = {Logarithmic Negativity: A Full Entanglement Monotone That is not Convex},
  author = {Plenio, M. B.},
  journal = {Phys. Rev. Lett.},
  volume = {95},
  issue = {9},
  pages = {090503},
  numpages = {4},
  year = {2005},
  month = {Aug},
  publisher = {American Physical Society},
  doi = {10.1103/PhysRevLett.95.090503},
  url = {https://link.aps.org/doi/10.1103/PhysRevLett.95.090503}
}

@article{finiteEnt2009,
  title = {Critical and noncritical long-range entanglement in Klein-Gordon fields},
  author = {Marcovitch, S. and Retzker, A. and Plenio, M. B. and Reznik, B.},
  journal = {Phys. Rev. A},
  volume = {80},
  issue = {1},
  pages = {012325},
  numpages = {4},
  year = {2009},
  month = {Jul},
  publisher = {American Physical Society},
  doi = {10.1103/PhysRevA.80.012325},
  url = {https://link.aps.org/doi/10.1103/PhysRevA.80.012325}
}

@article{KlcoUVIR,
  title = {Entanglement Spheres and a {UV-IR} Connection in Effective Field Theories},
  author = {Klco, Natalie and Savage, Martin J.},
  journal = {Phys. Rev. Lett.},
  volume = {127},
  issue = {21},
  pages = {211602},
  numpages = {7},
  year = {2021},
  month = {Nov},
  publisher = {American Physical Society},
  doi = {10.1103/PhysRevLett.127.211602},
  url = {https://link.aps.org/doi/10.1103/PhysRevLett.127.211602}
}

@article{KlcoEntAllDist,
  title = {Detecting spacelike vacuum entanglement at all distances and promoting negativity to a necessary and sufficient entanglement measure in many-body regimes},
  author = {Gao, Boyu and Klco, Natalie},
  journal = {Phys. Rev. A},
  volume = {112},
  issue = {1},
  pages = {012430},
  numpages = {13},
  year = {2025},
  month = {Jul},
  publisher = {American Physical Society},
  doi = {10.1103/m9w1-ppqz},
  url = {https://link.aps.org/doi/10.1103/m9w1-ppqz}
}

@article{embezzlers2024,
  title = {Relativistic Quantum Fields Are Universal Entanglement Embezzlers},
  author = {van Luijk, Lauritz and Stottmeister, Alexander and Werner, Reinhard F. and Wilming, Henrik},
  journal = {Phys. Rev. Lett.},
  volume = {133},
  issue = {26},
  pages = {261602},
  numpages = {8},
  year = {2024},
  month = {Dec},
  publisher = {American Physical Society},
  doi = {10.1103/PhysRevLett.133.261602},
  url = {https://link.aps.org/doi/10.1103/PhysRevLett.133.261602}
}

@article{Layden_Martin-Martinez_Kempf_2016, title={Universal scheme for indirect quantum control}, volume={93}, ISSN={2469-9926, 2469-9934}, DOI={10.1103/PhysRevA.93.040301}, number={4}, journal={Phys. Rev. A}, author={Layden, David and Martin-Martinez, Eduardo and Kempf, Achim}, year={2016}, pages={040301} }

@article{Martin-Martinez_Sutherland_2014, title={Quantum gates via relativistic remote control}, volume={739}, ISSN={0370-2693}, DOI={10.1016/j.physletb.2014.10.038}, journal={Phys. Lett. B}, author={Martín-Martínez, Eduardo and Sutherland, Chris}, year={2014}, pages={74–82} }

@article{Martin-Martinez_Aasen_Kempf_2013, title={Processing quantum information with relativistic motion of atoms}, volume={110}, ISSN={0031-9007, 1079-7114}, DOI={10.1103/PhysRevLett.110.160501}, number={16}, journal={Phys. Rev. Lett.}, author={Martin-Martinez, Eduardo and Aasen, David and Kempf, Achim}, year={2013}, pages={160501} }

@article{beyondPertDet2013,
  title = {Detectors for probing relativistic quantum physics beyond perturbation theory},
  author = {Brown, Eric G. and Mart\'{\i}n-Mart\'{\i}nez, Eduardo and Menicucci, Nicolas C. and Mann, Robert B.},
  journal = {Phys. Rev. D},
  volume = {87},
  issue = {8},
  pages = {084062},
  numpages = {19},
  year = {2013},
  month = {Apr},
  publisher = {American Physical Society},
  doi = {10.1103/PhysRevD.87.084062},
  url = {https://link.aps.org/doi/10.1103/PhysRevD.87.084062}
}

@article{PetarNonpertEH,
  title = {Nonperturbative analysis of entanglement harvesting from coherent field states},
  author = {Simidzija, Petar and Mart\'{\i}n-Mart\'{\i}nez, Eduardo},
  journal = {Phys. Rev. D},
  volume = {96},
  issue = {6},
  pages = {065008},
  numpages = {19},
  year = {2017},
  month = {Sep},
  publisher = {American Physical Society},
  doi = {10.1103/PhysRevD.96.065008},
  url = {https://link.aps.org/doi/10.1103/PhysRevD.96.065008}
}

@book{ReedSimon1,
  author    = {Reed, Michael and Simon, Barry},
  title     = {Methods of Modern Mathematical Physics. I: Functional Analysis},
  publisher = {Academic Press},
  address   = {New York},
  year      = {1980}
}

@book{Kato,
  author    = {Kato, Tosio},
  title     = {Perturbation Theory for Linear Operators},
  edition   = {2},
  publisher = {Springer},
  address   = {Berlin},
  year      = {1976}
}

@misc{SERCMEE,
  author = {Morote-Balboa, Marcos and Mart\'in-Mart\'inez, Eduardo and Rick Perche, T.},
  title = {(in preparation)}
}

@article{MO_RbK_loss,
  title = {Universal Three-Body Physics in Ultracold {KRb} Mixtures},
  author = {Wacker, L. J. and J\o{}rgensen, N. B. and Birkmose, D. and Winter, N. and Mikkelsen, M. and Sherson, J. and Zinner, N. and Arlt, J. J.},
  journal = {Phys. Rev. Lett.},
  volume = {117},
  issue = {16},
  pages = {163201},
  numpages = {6},
  year = {2016},
  month = {Oct},
  publisher = {American Physical Society},
  doi = {10.1103/PhysRevLett.117.163201},
  url = {https://link.aps.org/doi/10.1103/PhysRevLett.117.163201}
}

\appendix

\onecolumngrid 

\section{Computation of the Matrix elements}\label{app:matrix}
In this appendix, we derive Eq.~\eqref{eq:Pnm} and use it to compute the matrices of the different propagators at play, $G_\tc{ab}^{\sct{++}}$, $W_\tc{aa}^{\sct{-+}}$, $W_\tc{bb}^{\sct{-+}}$ and $\Delta_\tc{ab}^{\sct{++}}$. 

In Eq.~\eqref{eq:coefficients-polynomial-development}, we define $P_{nm}$ as:
\begin{equation}
    P_{nm} \equiv P ({\Lambda}^\pm_{i, n},{\Lambda}^\pm_{j, m}).
\end{equation}
Writing this explicitly,
\begin{align}\label{app:Pnm-app}
    P_{nm}&= \int \dd V \dd V' {\Lambda}^\pm_{i, n}(\xx) {\Lambda}^\pm_{j, m} (\xx') P(\xx, \xx')\nonumber \\ \nonumber \\
    &= \int \dd V \dd V' \ee^{\pm \ii \Omega t}{\Lambda}_{i, n}(\xx) \ee^{\pm \ii \Omega t'}{\Lambda}_{j, m} (\xx') P(\xx, \xx'), \quad  \text{using Eq.~\eqref{eq:lambda-ring-pm},}\nonumber \\ \nonumber \\
    & = \int \dd V \dd V' \ee^{\pm \ii \Omega t}\ee^{\pm \ii \Omega t'} h_n(t,T) h_m(t',T) F(\bm x) F(\bm x') P(\xx, \xx'),\quad \text{using ${\Lambda}_{i,n}(\mf x) = h_n(t,T)F(\bm x)$}\nonumber \\ \nonumber \\
    & = \int \dd V \dd V' \ee^{\pm \ii \Omega t}\ee^{\pm \ii \Omega t'} h_n(t,T) h_m(t',T)  \frac{\ee^{- \frac{|\bm x|^2}{2\sigma^2}}}{(2\pi \sigma^2)^{3/2}} \frac{\ee^{- \frac{|\bm x' - \bm L|^2}{2\sigma^2}}}{(2\pi \sigma^2)^{3/2}}P(\xx, \xx'),\quad \text{using Eq.~\eqref{eq:choice-F}}\nonumber \\ \nonumber \\
    & = \int \dd V \dd V' \ee^{\pm \ii \Omega t}\ee^{\pm \ii \Omega t'} \frac{\pi^{-1/2}}{ \sqrt{2^{n+m} n! m!}T} H_n(t/T) \, \ee^{-\frac{t^2}{2T^2}}  H_m(t'/T) \, \ee^{-\frac{t'^2}{2T^2}}  \frac{\ee^{- \frac{|\bm x|^2}{2\sigma^2}}}{(2\pi \sigma^2)^{3/2}} \frac{\ee^{- \frac{|\bm x' - \bm L|^2}{2\sigma^2}}}{(2\pi \sigma^2)^{3/2}}P(\xx, \xx'),\quad \text{using Eq.~\eqref{eq:h(t)}}\nonumber \\ \nonumber \\
    & = \frac{\pi^{-1/2}}{\sqrt{2^{n+m} n! m!}\,T} H_n \bigg(\frac{\dd}{\dd \alpha} \bigg) H_m \bigg(\frac{\dd}{\dd \beta} \bigg) \int \dd V \dd V' \ee^{\pm \ii \Omega t}\ee^{\pm \ii \Omega t'} \ee^{\alpha t} \ee^{\beta t} \ee^{-\frac{t^2}{2T^2}}   \ee^{-\frac{t'^2}{2T^2}}  \frac{\ee^{- \frac{|\bm x|^2}{2\sigma^2}}}{(2\pi \sigma^2)^{3/2}} \frac{\ee^{- \frac{|\bm x' - \bm L|^2}{2\sigma^2}}}{(2\pi \sigma^2)^{3/2}}P(\xx, \xx')\nonumber \\ \nonumber \\
    &\equiv  \frac{\pi^{-1/2}}{\sqrt{2^{n+m} n! m!}\, T} H_n \bigg(\frac{\dd}{\dd \alpha} \bigg) H_m \bigg(\frac{\dd}{\dd \beta} \bigg) P (\Lambda_\alpha^\pm,\Lambda_\beta^\pm),
\end{align}
where we used Eq.~\eqref{eq:feynman-trick} in the last equality. The next step is to then find $P (\Lambda_\alpha^\pm,\Lambda_\beta^\pm)$ for the propagators we are interested in,  $H (\Lambda_\alpha^+,\Lambda_\beta^+)$, $\Delta (\Lambda_\alpha^+,\Lambda_\beta^+)$ and $W (\Lambda_\alpha^-,\Lambda_\beta^+)|_{\bm L = 0}$, and then compute their different derivatives with respect to $\alpha$ and $\beta$, which we will do in App.~\ref{app:leibniz}.
To find the corresponding expressions, we use the results of~\cite{closedform2024}, where, with the spatial and temporal smearing functions
\begin{equation}\label{eq:rick-switching}
    F_i(\bm x) =  \frac{\ee^{- \frac{|\bm x - \bm x_i|^2}{2\sigma_i^2}}}{(2\pi \sigma^2)^{3/2}}, \quad \quad \chi_i(t) = \ee^{- \frac{(t-t_i)^2}{2T_i^2}} \ee^{\ii \Omega_i t},
\end{equation}
the following closed-form expressions for the Hadamard, symmetric and Wightman propagators were found:
\begin{align}
    H(f_1,f_2) &= \frac{T_1T_2\ee^{\ii (\Omega_1 t_1+\Omega_2 t_2)}\ee^{- \frac{\Omega_2^2 T_2^2}{2} - \frac{\Omega_1^2 T_1^2}{2}}}{2\sqrt{2\pi}|\bm L|\sqrt{T_1^2 + T_2^2 + \sigma_1^2 + \sigma_2^2}} \Bigg[\ee^{- \frac{(t_0 - |\bm L|+\ii (\Omega_1 T_1^2- \Omega_2 T_2^2))^2}{2(T_1^2 + T_2^2 + \sigma_1^2 + \sigma_2^2)}}\text{erfi}\left(\frac{|\bm L|-t_0-\ii (\Omega_1 T_1^2- \Omega_2 T_2^2)}{\sqrt{2}\sqrt{T_1^2 + T_2^2 + \sigma_1^2 + \sigma_2^2}}\right)\nonumber\\
     &\:\:\:\:\:\:\:\:\:\:\:\:\:\:\:\:\:\:\:\:\:\:\:\:\:\:\:\:\:\:\:\:\:\:\:\:\:\:\:\:\:\:\:\:\:\:\:\:\:\:\:\:\:\:\:\:\:\:\:\:\:\:\:\:\:\:\:\:\:\:\:\:\:\:\:\:\:\:\:+\ee^{- \frac{(t_0 + |\bm L|+\ii (\Omega_1 T_1^2- \Omega_2 T_2^2))^2}{2(T_1^2 + T_2^2 + \sigma_1^2 + \sigma_2^2)}}\text{erfi}\left(\frac{|\bm L|+t_0+\ii (\Omega_1 T_1^2- \Omega_2 T_2^2)}{\sqrt{2}\sqrt{T_1^2 + T_2^2 + \sigma_1^2 + \sigma_2^2}}\right)\Bigg],\\ \nonumber\\
    \Delta(f_1,f_2) &= - \frac{T_1T_2\ee^{\ii (\Omega_1 t_1+\Omega_2 t_2)}\ee^{- \frac{\Omega_2^2 T_2^2}{2} - \frac{\Omega_1^2 T_1^2}{2}}}{2\sqrt{2\pi}|\bm L|\sqrt{T_1^2 + T_2^2 + 2\sigma^2}}\Bigg[\ee^{- \frac{(t_0 - |\bm L|+\ii (\Omega_1 T_1^2- \Omega_2 T_2^2))^2}{2(T_1^2 + T_2^2 + 2\sigma^2)}} \text{erf}\left(\tfrac{|\bm L|(T_1^2+T_2^2)+2\sigma^2(t_0+\ii (\Omega_1 T_1^2- \Omega_2 T_2^2))}{2\sigma\sqrt{T_1^2 + T_2^2}\sqrt{T_1^2 + T_2^2 + 2\sigma^2}}\right) \nonumber\\
    &\:\:\:\:\:\:\:\:\:\:\:+\ee^{- \frac{(t_0 + |\bm L|+\ii (\Omega_1 T_1^2- \Omega_2 T_2^2))^2}{2(T_1^2 + T_2^2 + 2\sigma^2)}}\text{erf}\left(\tfrac{|\bm L|(T_1^2+T_2^2) - 2\sigma^2(t_0+\ii (\Omega_1 T_1^2- \Omega_2 T_2^2))}{2\sigma\sqrt{T_1^2 + T_2^2}\sqrt{T_1^2 + T_2^2 + 2\sigma^2}}\right)\Bigg],\\ \nonumber\\
    W(f_1,f_2)&= \frac{T_1T_2\ee^{\ii (\Omega_1 t_1+\Omega_2 t_2)}\ee^{- \frac{\Omega_2^2 T_2^2}{2} - \frac{\Omega_1^2 T_1^2}{2}}}{4\sqrt{2\pi}|\bm L|\sqrt{T_1^2 + T_2^2 + \sigma_1^2 + \sigma_2^2}} \Bigg[\ee^{- \frac{(t_0 - |\bm L|+\ii (\Omega_1 T_1^2- \Omega_2 T_2^2))^2}{2(T_1^2 + T_2^2 + \sigma_1^2 + \sigma_2^2)}}\Bigg(-\ii - \text{erfi}\left(\tfrac{t_0 - |\bm L|+\ii (\Omega_1 T_1^2- \Omega_2 T_2^2)}{\sqrt{2}\sqrt{T_1^2 + T_2^2 + \sigma_1^2 + \sigma_2^2}}\right)\Bigg)\nonumber\\
     &\:\:\:\:\:\:\:\:\:\:\:\:\:\:\:\:\:\:\:\:\:\:\:\:\:\:\:\:\:\:\:\:\:\:\:\:\:\:\:\:\:\:\:\:\:\:\:\:\:\:\:\:\:\:\:\:\:\:\:\:\:\:\:\:\:\:\:\:\:\:\:\:\:\:\:\:\:\:\:+\ee^{- \frac{(t_0 + |\bm L|+\ii (\Omega_1 T_1^2- \Omega_2 T_2^2))^2}{2(T_1^2 + T_2^2 + \sigma_1^2 + \sigma_2^2)}}\Bigg(\ii + \text{erfi}\left(\tfrac{t_0 + |\bm L|+\ii (\Omega_1 T_1^2- \Omega_2 T_2^2)}{\sqrt{2}\sqrt{T_1^2 + T_2^2 + \sigma_1^2 + \sigma_2^2}}\right)\Bigg)\Bigg],
\end{align}
where $\bm L = \bm x_1 - \bm x_2$ and $t_0 = t_1 - t_2$.
Our goal now is to recover switching functions of the form used in Eq.~\eqref{eq:Pnm} from Eq.~\eqref{eq:rick-switching}. First, let us note that, in our case, we are considering identical point-like detectors that interact simultaneously in their comoving frame, and we can therefore set $\sigma_i  = 0$, $\bm x_1 = 0$, $\bm x_2 = \bm L$, and $t_i =  0$. We introduce the same time scale for the interactions by setting $T_i = T$. Then, for each propagator, we can recover the exponential factors $\ee^{\pm\alpha t},\ee^{\pm\beta t}$ by picking $\Omega_1$ and $\Omega_2$ appropriately. Indeed, to recover $H(\Lambda_\alpha^+,\Lambda_\beta^+)$ and $\Delta(\Lambda_\alpha^+,\Lambda_\beta^+)$, we set $\Omega_1 \equiv \Omega - \ii \alpha, \Omega_2 \equiv \Omega - \ii \beta$, and $W(\Lambda_\alpha^-,\Lambda_\beta^+)$ is obtained with $\Omega_1 \equiv -\Omega - \ii \alpha$ , $\Omega_2 \equiv \Omega - \ii \beta$, which gives
\begin{align}
   H (\Lambda_\alpha^+,\Lambda_\beta^+)&= \frac{T \ee^{-\Omega^2 T^2}}{4 |\bm L| \sqrt{\pi}} \ee^{\frac{1}{2} T^2 (\alpha^2+\beta^2)} \ee^{-\ii \Omega T^2 (\alpha+\beta)} \Bigg[\ee^{-\frac{(|\bm L|+(\alpha-\beta)T^2)^2}{4 T^2}} \text{erfi} \bigg(\frac{|\bm L|+(\alpha-\beta)^2 T^2}{2 T} \bigg)\nonumber \\
    &\quad \quad \quad \quad \quad \quad \quad \quad \quad+ \ee^{-\frac{(|\bm L|+(-\alpha+\beta)T^2)^2}{4 T^2}} \text{erfi} \bigg(\frac{|\bm L|+(-\alpha+\beta) T^2}{2 T} \bigg)\Bigg],\label{eq:generator-H-a-b} \\ \nonumber \\
\Delta (\Lambda_\alpha^+,\Lambda_\beta^+) &=  -\frac{T \ee^{-\frac{|\bm L|^2}{4 T^2}+T^4(\alpha+\beta-2\ii \Omega)^2}}{4 |\bm L| \sqrt{\pi}} \bigg(\ee^{\frac{1}{2} (\alpha - \beta) |\bm L| } + \ee^{-\frac{1}{2} (\alpha - \beta) |\bm L| }\bigg),\label{eq:generator-Delta-a-b}  \\ \nonumber \\
 W (\Lambda_\alpha^-,\Lambda_\beta^+)|_{\bm L = 0} &= \dfrac{ \ee^{ - T^2 \Omega^2}}{4 \pi} \ee^{\frac{1}{2} T^2(\alpha^2+\beta^2)}\ee^{\ii T^2 \Omega (\alpha-\beta)}  - \dfrac{T(\Omega-\ii \frac{\alpha-\beta}{2}) }{4\sqrt{\pi}} \ee^{\frac{1}{4} T^2 (\alpha+\beta)^2} \Bigg[1 - \text{erf} \bigg( T \Omega - \ii T \frac{\alpha-\beta}{2}\bigg) \Bigg].\label{eq:generator-W-a-b}
\end{align}

\section{Reducing the expressions for derivatives of products}\label{app:leibniz}

Even though the individual computations of the derivatives that give rise to the matrix components $P_{nm}$ can be done in closed form, implementing them numerically remains a computationally expensive task. In this appendix, we express the derivatives of $P(\alpha,\beta) = P(\Lambda_\alpha^\pm,\Lambda_\beta^\pm)$ in terms of derivatives of simpler functions that can be stored and efficiently computed numerically to reconstruct the matrix $P_{nm}$ as given by Eq.~\eqref{app:Pnm-app}.

We first perform the change of variables
\begin{equation}\label{eq:u-v-a-b}
    u = \alpha + \beta, \quad v = \alpha - \beta
\end{equation}
so that
\begin{equation}\label{eq:a-b-u-b}
    \alpha = \frac{u+v}{2}, \quad \beta = \frac{u-v}{2}, \quad \alpha^2+\beta^2 = \frac{u^2+v^2}{2}.
\end{equation}
We can then rewrite Eqs.~\eqref{eq:generator-H-a-b}-~\eqref{eq:generator-W-a-b} in terms of these new variables. \\
Let us start with $H$:
\begin{align}\label{eq:generators-u-v}
   H (u, v)&= \frac{T \ee^{-\Omega^2 T^2}}{4 L \sqrt{\pi}} \ee^{\frac{1}{4} T^2 (u^2+v^2)} \ee^{-\ii \Omega T^2 u} \Bigg[\ee^{-\frac{(L+vT^2)^2}{4 T^2}} \text{erfi} \bigg(\frac{L+v T^2}{2 T} \bigg)+ \ee^{-\frac{(L-vT^2)^2}{4 T^2}} \text{erfi} \bigg(\frac{L-v T^2}{2 T} \bigg)\Bigg] \nonumber  \\ \nonumber  \\ 
   &= \frac{T \ee^{-\Omega^2 T^2}}{4 L \sqrt{\pi}} \ee^{\frac{1}{4} T^2 u^2} \ee^{-\ii \Omega T^2 u} \Bigg[ \ee^{\frac{1}{4} T^2 v^2}\ee^{-\frac{(L+vT^2)^2}{4 T^2}} \text{erfi} \bigg(\frac{L+v T^2}{2 T} \bigg)+\ee^{\frac{1}{4} T^2 v^2} \ee^{-\frac{(L-vT^2)^2}{4 T^2}} \text{erfi} \bigg(\frac{L-v T^2}{2 T} \bigg)\Bigg] \nonumber  \\ \nonumber  \\
   &= F_H(u) [G_H(v)+G_H(-v)],
\end{align}
where we have defined:
\begin{equation}\label{eq:def-F-G}
    F_H(u)= \frac{T \ee^{\frac{1}{4} T^2 (u-2\ii \Omega)^2}}{4 L \sqrt{\pi}}, \quad G_H(v)=\ee^{\frac{-L^2}{4 T^2}}\ee^{-\frac{L v}{2}} \text{erfi} \bigg(\frac{L+v T^2}{2 T} \bigg).
\end{equation}
The different derivatives of $H(\alpha,\beta)$ are then given by:
\begin{equation}
    \frac{\partial^{i+j}}{\partial \alpha^i \partial \beta^j} H(\alpha,\beta) = \bigg( \frac{\partial}{\partial u} + \frac{\partial}{\partial v}\bigg)^i \bigg( \frac{\partial}{\partial u} - \frac{\partial}{\partial v}\bigg)^j F_H(u) [G_H(v)+G_H(-v)].
\end{equation}
We can now use the Leibniz rule to see how the operator acts on $F_H(u) G_H(v)$:
\begin{align}
   & \bigg( \frac{\partial}{\partial u} + \frac{\partial}{\partial v}\bigg)^i \bigg( \frac{\partial}{\partial u} - \frac{\partial}{\partial v}\bigg)^j \bigg( F_H(u) G_H(v) \bigg) = \bigg( \frac{\partial}{\partial u} + \frac{\partial}{\partial v}\bigg)^i \Bigg\{\bigg( \frac{\partial}{\partial u} - \frac{\partial}{\partial v}\bigg)^j \bigg( F_H(u) G_H(v) \bigg) \Bigg\} \nonumber \\ \nonumber \\
    &=\bigg( \frac{\partial}{\partial u} + \frac{\partial}{\partial v}\bigg)^i \Bigg\{ \sum_{p=0}^j \binom{j}{p} \bigg[ \bigg( \frac{\partial}{\partial u} - \frac{\partial}{\partial v}\bigg)^{j-p} F_H(u)\bigg] \bigg[ \bigg( \frac{\partial}{\partial u} - \frac{\partial}{\partial v}\bigg)^{p} G_H(v)\bigg] \Bigg\} \nonumber \\ \nonumber \\
    &= \bigg( \frac{\partial}{\partial u} + \frac{\partial}{\partial v}\bigg)^i \Bigg\{ \sum_{p=0}^j \binom{j}{p} \bigg[ \sum_{r=0}^{j-p} \binom{j-p}{r} (-1)^r \bigg(\frac{\partial}{\partial u} \bigg)^{j-p-r} \bigg(\frac{\partial}{\partial v} \bigg)^{r} F_H(u)\bigg] \bigg[ \sum_{s=0}^{p} \binom{p}{s} (-1)^s \bigg(\frac{\partial}{\partial u} \bigg)^{p-s} \bigg(\frac{\partial}{\partial v} \bigg)^{s} G_H(v)\bigg]  \Bigg\} \nonumber \\ \nonumber \\
    &= \bigg( \frac{\partial}{\partial u} + \frac{\partial}{\partial v}\bigg)^i \Bigg\{ \sum_{p=0}^j \binom{j}{p} \bigg[ \sum_{r=0}^{j-p} \binom{j-p}{r} (-1)^r \bigg(\frac{\partial}{\partial u} \bigg)^{j-p-r}\! \underbrace{\bigg(\frac{\partial}{\partial v} \bigg)^{r} F_H(u)}_{\not=0 \iff r=0}\bigg] \bigg[ \sum_{s=0}^{p} \binom{p}{s} (-1)^s  \bigg(\frac{\partial}{\partial v} \bigg)^{s} \!\!\underbrace{\bigg(\frac{\partial}{\partial u} \bigg)^{p-s} \!G_H(v)}_{\substack{\not=0 \iff p-s=0 \\ \iff s=p}}\bigg]  \Bigg\} \nonumber \\ \nonumber \\
    &= \bigg( \frac{\partial}{\partial u} + \frac{\partial}{\partial v}\bigg)^i \Bigg\{ \sum_{p=0}^j \binom{j}{p} \bigg[ \bigg(\frac{\partial}{\partial u} \bigg)^{j-p}  F_H(u)\bigg] \bigg[  (-1)^p  \bigg(\frac{\partial}{\partial v} \bigg)^{p} G_H(v)\bigg]  \Bigg\} \nonumber \\ \nonumber \\
    &=   \sum_{p=0}^j \binom{j}{p} \bigg( \frac{\partial}{\partial u} + \frac{\partial}{\partial v}\bigg)^i \bigg[ \bigg(\frac{\partial}{\partial u} \bigg)^{j-p}  F_H(u)\bigg] \bigg[  (-1)^p  \bigg(\frac{\partial}{\partial v} \bigg)^{p} G_H(v)\bigg]   \nonumber \\ \nonumber \\
    &=\sum_{p=0}^j \binom{j}{p} \sum_{q=0}^i \binom{i}{q} \Bigg\{ \bigg( \frac{\partial}{\partial u} + \frac{\partial}{\partial v}\bigg)^{i-q} \bigg[ \bigg(\frac{\partial}{\partial u} \bigg)^{j-p}  F_H(u)\bigg] \Bigg\} \Bigg\{ \bigg[  (-1)^p \bigg( \frac{\partial}{\partial u} + \frac{\partial}{\partial v}\bigg)^{q} \bigg(\frac{\partial}{\partial v} \bigg)^{p} G_H(v)\bigg]  \Bigg\} \nonumber \\ \nonumber \\
    &=\sum_{p=0}^j \sum_{q=0}^i \binom{j}{p} \binom{i}{q} (-1)^p \Bigg\{ \bigg[\sum_{r=0}^{i-q} \binom{i-q}{r} \bigg( \frac{\partial}{\partial u}\bigg)^{i-q-r} \bigg( \frac{\partial}{\partial v}\bigg)^{r} \bigg] \bigg[ \bigg(\frac{\partial}{\partial u} \bigg)^{j-p}  F_H(u)\bigg]\Bigg\} \nonumber \\ 
    & \quad \quad \quad \quad \quad \quad \quad \quad \quad \quad \quad \quad \quad \quad \quad \times  \left\{ \bigg[\sum_{s=0}^{q} \binom{q}{s} \bigg( \frac{\partial}{\partial u}\bigg)^{q-s} \bigg( \frac{\partial}{\partial v}\bigg)^{s} \bigg] \bigg[ \bigg(\frac{\partial}{\partial v} \bigg)^{p}  G_H(v)\bigg]\right\} \nonumber \\ \nonumber \\
    &=\sum_{p=0}^j \sum_{q=0}^i \binom{j}{p} \binom{i}{q} (-1)^p  \bigg[\sum_{r=0}^{i-q} \binom{i-q}{r} \bigg( \frac{\partial}{\partial u}\bigg)^{i+j-p-q-r} \bigg( \frac{\partial}{\partial v}\bigg)^{r} F_H(u)\bigg] \bigg[\sum_{s=0}^{q} \binom{q}{s} \bigg( \frac{\partial}{\partial u}\bigg)^{q-s} \bigg( \frac{\partial}{\partial v}\bigg)^{p+s}  G_H(v)\bigg] \nonumber \\ \nonumber \\
    &=\sum_{p=0}^j \sum_{q=0}^i \binom{j}{p} \binom{i}{q} (-1)^p  \bigg[\sum_{r=0}^{i-q} \binom{i-q}{r} \bigg( \frac{\partial}{\partial u}\bigg)^{i+j-p-q-r} \underbrace{\bigg( \frac{\partial}{\partial v}\bigg)^{r} F_H(u)}_{\not = 0 \iff r=0}\bigg] \bigg[\sum_{s=0}^{q} \binom{q}{s}  \bigg( \frac{\partial}{\partial v}\bigg)^{p+s} \underbrace{\bigg( \frac{\partial}{\partial u}\bigg)^{q-s}  G_H(v)}_{\not = 0 \iff s=q}\bigg] \nonumber \\ \nonumber \\
    &=\sum_{p=0}^j \sum_{q=0}^i \binom{j}{p} \binom{i}{q} (-1)^p  \bigg[ \bigg( \frac{\partial}{\partial u}\bigg)^{i+j-p-q}  F_H(u)\bigg] \bigg[  \bigg( \frac{\partial}{\partial v}\bigg)^{p+q}   G_H(v)\bigg].
\end{align}
At this stage, one only requires knowledge of the derivatives of the functions $F_H(u)$ and $G_H(v)$ to obtain the general derivatives of $H(\alpha,\beta)$.

Following the same procedure for $\Delta$ and $W$:
\begin{align}
     \frac{\partial^{i+j}}{\partial \alpha^i \partial \beta^j} \Delta(\alpha,\beta) &= \sum_{p=0}^i \sum_{q=0}^j \binom{i}{p} \binom{j}{q} (-1)^q (1+(-1)^{p+q} \bigg[\bigg( \frac{\partial}{\partial u}\bigg)^{i+j-p-q} F_\Delta (u) \bigg] \bigg[\bigg( \frac{\partial}{\partial v}\bigg)^{p+q} G_\Delta (v) \bigg], \\ \nonumber \\
     \frac{\partial^{i+j}}{\partial \alpha^i \partial \beta^j} W(\alpha,\beta) &= \sum_{p=0}^i \sum_{q=0}^j \binom{i}{p} \binom{j}{q} (-1)^q \bigg[\bigg( \frac{\partial}{\partial u}\bigg)^{i+j-p-q} F_W (u) \bigg] \bigg[\bigg( \frac{\partial}{\partial v}\bigg)^{p+q} G_W (v) \bigg],
\end{align}
with:
\begin{align}
    F_\Delta(u) &= -\frac{T \ee^{-\frac{L^2}{4T^2}}}{4L\sqrt{\pi}} \ee^{\frac{1}{4}T^2(u-2\ii \Omega)^2}, \quad G_\Delta(v) = \ee^{\frac{Lv}{2}}, \\ \nonumber \\
    F_W(u)&= \frac{\ee^{\frac{1}{4}T^2 u^2}}{4 \pi}, \quad G_W(v) = \ee^{\frac{1}{4} T^2 (v+2\ii \Omega)^2} - T\sqrt{\pi} \bigg(\Omega - \frac{1}{2} \ii v\bigg) \bigg[ 1 - \text{erf}\bigg(T \Omega - \frac{1}{2} \ii v\bigg) \bigg].
\end{align}
For the sake of simplicity, we will introduce the dimensionless variables:
\begin{equation}
    \mu = u T, \quad \nu = v T, \quad\w =\Omega T, \quad \ell=\frac{L}{T},
\end{equation}
and the different functions become:
\begin{align}
    F_H(u)&= \frac{ \ee^{\frac{1}{4} (\mu-2\ii \omega)^2}}{4 \ell \sqrt{\pi}}, \quad G_H(v)=\ee^{-\frac{1}{4} \ell^2}\ee^{-\frac{1}{2}\ell \nu} \text{erfi} \bigg(\frac{1}{2}(\ell + \nu) \bigg), \\ \nonumber \\
    F_\Delta(u) &= -\frac{ \ee^{-\frac{1}{4} \mu^2}}{4 \ell\sqrt{\pi}} \ee^{\frac{1}{4}(\mu-2\ii \omega)^2}, \quad G_\Delta(v) = \ee^{\frac{1}{2} \ell \nu}, \\ \nonumber \\
    F_W(u)&= \frac{\ee^{\frac{1}{4}\mu^2}}{4 \pi}, \quad G_W(v) = \ee^{\frac{1}{4}  (\nu+2\ii \omega)^2} - \sqrt{\pi} \bigg(\omega - \frac{1}{2} \ii \nu\bigg) \bigg[ 1 - \text{erf}\bigg( \omega - \frac{1}{2} \ii \nu\bigg) \bigg].
\end{align}
It is important to note that this change of variable also affects the differential operators:
\begin{equation}
    \bigg( \frac{\partial}{\partial u} \bigg)^k = T^k \bigg( \frac{\partial}{\partial \mu} \bigg)^k, \quad \bigg( \frac{\partial}{\partial v} \bigg)^k = T^k \bigg( \frac{\partial}{\partial \nu} \bigg)^k.
\end{equation}
Finally, we can write:
\begin{align}
    G_{nm}&=   H_n \bigg( \dfrac{\dd}{\dd \alpha} \bigg) \,  H_m \bigg( \dfrac{\dd}{\dd \beta} \bigg)  H (\ee^{\alpha t} \Lambda^+_{\tc{a}},\ee^{\beta t'} \Lambda^+_{\tc{b}})\Bigg\lvert_{\alpha=\beta=0} \nonumber \\ \nonumber\\
   &=   \Bigg[ n! \sum_{r=0}^{\lfloor \frac{n}{2}\rfloor} \frac{(-1)^r}{r!(n-2r)!} \bigg(2 \frac{d}{d \alpha}\bigg)^{n-2r} \Bigg] \, \Bigg[ m! \sum_{s=0}^{\lfloor \frac{m}{2}\rfloor} \frac{(-1)^s}{s!(m-2s)!} \bigg(2 \frac{d}{d \beta}\bigg)^{m-2s} \Bigg] H (\ee^{\alpha t} \Lambda^+_{\tc{a}},\ee^{\beta t'} \Lambda^+_{\tc{b}})\Bigg\lvert_{\alpha=\beta=0}  \nonumber \\ \nonumber\\
   &= \sum_{r=0}^{\lfloor \frac{n}{2}\rfloor} \sum_{s=0}^{\lfloor \frac{m}{2}\rfloor} \frac{(-1)^{r+s}\, 2^{n+m-2(r+s)}}{r!s!(n-2r)!(m-2s)!} \bigg( \frac{d}{d \alpha}\bigg)^{n-2r} \bigg( \frac{d}{d \beta}\bigg)^{m-2s} H (\ee^{\alpha t} \Lambda^+_{\tc{a}},\ee^{\beta t'} \Lambda^+_{\tc{b}})\Bigg\lvert_{\alpha=\beta=0} 
    \nonumber \\ \nonumber\\
    &=\sum_{r=0}^{\lfloor \frac{n}{2}\rfloor} \sum_{s=0}^{\lfloor \frac{m}{2}\rfloor}  \frac{(-1)^{r+s}\, 2^{n+m-2(r+s)}}{r!s!(n-2r)!(m-2s)!} \sum_{i=0}^{n-2r} \sum_{j=0}^{m-2s} \binom{n-2r}{i} \binom{m-2s}{j}(-1)^j (1+(-1)^{i+j}) \nonumber \\ 
    &\quad  \times T^{n+m-2(r+s)}  \Bigg[ \bigg(  \frac{d}{d \mu}\bigg)^{n+m-2(r+s)-i-j} \bigg(  \frac{ \ee^{\frac{1}{4} (\mu-2\ii \w)^2}}{4 \ell \sqrt{\pi}}\bigg)\Bigg]\Bigg\lvert_{\mu=0}  \Bigg[ \bigg(  \frac{d}{d \nu}\bigg)^{i+j} \bigg(\ee^{-\frac{1}{4} \ell^2} \ee^{-\frac{1}{2} \ell \nu} \text{erfi} \bigg(\frac{1}{2} (\ell + \nu) \bigg) \bigg)\Bigg]\Bigg\lvert_{\nu=0} , \label{eq:Gnm} \\ 
    W_{nm} &= \sum_{r=0}^{\lfloor \frac{n}{2}\rfloor} \sum_{s=0}^{\lfloor \frac{m}{2}\rfloor}  \frac{(-1)^{r+s}\, 2^{n+m-2(r+s)}}{r!s!(n-2r)!(m-2s)!} \sum_{i=0}^{n-2r} \sum_{j=0}^{m-2s} \binom{n-2r}{i} \binom{m-2s}{j}(-1)^j (1+(-1)^{i+j}) \,T^{n+m-2(r+s)} \nonumber \\ 
    & \times   \bigg[ \bigg( \frac{d}{d \mu}\bigg)^{n+m-2(r+s)-i-j} \bigg( \frac{\ee^{\frac{1}{4} \mu^2}}{4 \pi}\bigg)\bigg]\bigg\lvert_{\mu=0}  \bigg[ \bigg( \frac{d}{d \nu}\bigg)^{i+j} \bigg(\ee^{\frac{1}{4}  (\nu+2\ii \w)^2} -  \sqrt{\pi} \bigg(\w - \ii \frac{1}{2}\nu\bigg) \bigg[ 1 - \text{erf} \bigg( \w - \ii \frac{1}{2}\nu\bigg)\bigg] \bigg)\bigg]\bigg\lvert_{\nu=0}, \label{eq:Wnm} \\ 
     \Delta_{nm} &= \sum_{r=0}^{\lfloor \frac{n}{2}\rfloor} \sum_{s=0}^{\lfloor \frac{m}{2}\rfloor}  \frac{(-1)^{r+s}\, 2^{n+m-2(r+s)}}{r!s!(n-2r)!(m-2s)!} \sum_{i=0}^{n-2r} \sum_{j=0}^{m-2s} \binom{n-2r}{i} \binom{m-2s}{j}(-1)^j (1+(-1)^{i+j}) \nonumber \\ 
    & \times T^{n+m-2(r+s)}  \bigg[ \bigg( \frac{d}{d \mu}\bigg)^{n+m-2(r+s)-i-j} \bigg( -\frac{ \ee^{-\frac{1}{4} \mu^2}}{4 \ell\sqrt{\pi}} \ee^{\frac{1}{4}(\mu-2\ii \omega)^2}\bigg)\bigg]\bigg\lvert_{\mu=0}  \bigg[ \bigg( \frac{d}{d \nu}\bigg)^{i+j} \bigg(\ee^{\frac{1}{2} \ell \nu} \bigg)\bigg]\bigg\lvert_{\nu=0}. \label{eq:Deltanm}
\end{align}

\section{Proof of convergence of the rescaled Hermite expansion}
\label{app:proof}

In this appendix we prove that any function in $L^2([-1,1])$ can be expanded in terms of rescaled Hermite functions, as in Eq.~\eqref{eq:chi}. Let
\begin{equation}
    h_n(x)
    =
    \frac{1}{\pi^{1/4}\sqrt{2^n n!}}
    H_n(x)\ee^{-x^2/2}
\end{equation}
be the standard normalized Hermite functions on $L^2(\mathbb R)$, and define
their rescalings by
\begin{equation}
    h_n^{(\sigma)}(x)
    =
    \sigma^{-1/2}h_n(x/\sigma).
\end{equation}
For each fixed $\sigma>0$, the family
$\{h_n^{(\sigma)}\}_{n\geq 0}$ is an orthonormal basis of
$L^2(\mathbb R)$. We shall consider the $N$-dependent scale
\begin{equation}
    \sigma_N
    =
    \frac{1}{\sqrt{2N+1}},
    \qquad
    \varepsilon_N:=\sigma_N^2
    =
    \frac{1}{2N+1}.
\end{equation}
Given $f\in L^2(\mathbb R)$, define
\begin{equation}
    f_N(x)
    =
    \sum_{n=0}^N
    \big\langle h_n^{(\sigma_N)},f\big\rangle
    h_n^{(\sigma_N)}(x).
    \label{eq:def_f_N}
\end{equation}
Equivalently, $f_N$ is the orthogonal projection of $f$ onto
\begin{equation}
    \mathcal H_N
    =
    \operatorname{span}\{h_0^{(\sigma_N)},\ldots,h_N^{(\sigma_N)}\}.
\end{equation}

The main result is the following proposition, where we denote $\mathbf{1}_{[a,b]}$ as the characteristic function of the interval $[a,b]$,
\begin{equation}
    \mathbf{1}_{[a,b]}(x) = \begin{cases}
        1, x\in [a,b],\\
        0, x\notin [a,b].
    \end{cases}
\end{equation}

\begin{proposition}
\label{prop:semiclassical_projection}
For every $f\in L^2(\mathbb R)$,
\begin{equation}
    f_N
    \longrightarrow
    \mathbf 1_{[-1,1]}f
    \qquad
    \text{in }L^2(\mathbb R).
    \label{eq:f_N_limit}
\end{equation}
In particular, the sequence $\{f_N\}_{N\geq 0}$ converges in
$L^2(\mathbb R)$ and its limit is equal to $f$ if and only if
$f$ is supported in $[-1,1]$.
\end{proposition}

\begin{proof}
We write the finite Hermite projection as a spectral projection of a
semiclassical harmonic oscillator. For $\varepsilon>0$, let
\begin{equation}
    H_\varepsilon
    =
    -\varepsilon^2\frac{d^2}{dx^2}+x^2
    \label{eq:H_epsilon}
\end{equation}
be the self-adjoint operator on $L^2(\mathbb R)$ associated with the closed,
nonnegative quadratic form
\begin{equation}
    q_\varepsilon[u]
    =
    \varepsilon^2
    \int_{\mathbb R}|u'(x)|^2\,dx
    +
    \int_{\mathbb R}x^2|u(x)|^2\,dx,
    \label{eq:q_epsilon}
\end{equation}
with form domain
\begin{equation}
    \mathcal Q
    =
    H^1(\mathbb R)
    \cap
    \{u\in L^2(\mathbb R): xu\in L^2(\mathbb R)\}.
\end{equation}
The operator $H_\varepsilon$ is nonnegative and self-adjoint by the
Friedrichs representation theorem for closed semibounded quadratic forms
\cite{ReedSimon1,Kato}.

The standard Hermite functions satisfy
\begin{equation}
    \left(
        -\frac{d^2}{dy^2}+y^2
    \right)h_n(y)
    =
    (2n+1)h_n(y).
\end{equation}
Setting $y=x/\sigma$ and $\varepsilon=\sigma^2$, one obtains
\begin{equation}
    H_\varepsilon h_n^{(\sigma)}
    =
    (2n+1)\varepsilon\,h_n^{(\sigma)}.
    \label{eq:scaled_hermite_eigenvalue}
\end{equation}
In particular, for $\varepsilon=\varepsilon_N=(2N+1)^{-1}$,
\begin{equation}
    H_{\varepsilon_N}h_n^{(\sigma_N)}
    =
    \frac{2n+1}{2N+1}h_n^{(\sigma_N)}.
\end{equation}
Hence
\begin{equation}
    \frac{2n+1}{2N+1}\leq 1
    \qquad\Longleftrightarrow\qquad
    n\leq N.
\end{equation}
Therefore the projection in Eq.~\eqref{eq:def_f_N} is precisely the spectral
projection of $H_{\varepsilon_N}$ onto the interval $(-\infty,1]$:
\begin{equation}
    f_N
    =
    \mathbf 1_{(-\infty,1]}(H_{\varepsilon_N})f.
    \label{eq:f_N_spectral_projection}
\end{equation}

We now identify the strong limit of these spectral projections. Let $M$ be
the multiplication operator
\begin{equation}
    (Mu)(x)=x^2u(x),
\end{equation}
with domain
\begin{equation}
    D(M)=\{u\in L^2(\mathbb R): x^2u\in L^2(\mathbb R)\}.
\end{equation}
We will first prove that
\begin{equation}
    H_\varepsilon
    \longrightarrow
    M
    \qquad
    \text{in the strong resolvent sense as }\varepsilon\to0.
    \label{eq:strong_resolvent_claim}
\end{equation}

Recall the following standard criterion. If $A_j$ and $A$ are self-adjoint
operators and $z$ is a point in the common resolvent set, if we have that for each $g\in L^2(\mathbb R)$,
\begin{equation}
   \qquad(A_j-z)^{-1}g
    \longrightarrow
    (A-z)^{-1}g,    
\end{equation}
then $A_j\to A$ in the strong resolvent sense. This property can also be used to characterize the strong resolvent convergence~\cite{ReedSimon1,Kato}. Since $H_\varepsilon\geq0$ and
$M\geq0$, we consider the point $z=-1$, which lies outside their spectrum and within the resolvent set.

Define
\begin{equation}
    R_\varepsilon
    =
    (H_\varepsilon+1)^{-1},
    \qquad
    R_0
    =
    (M+1)^{-1}.
\end{equation}
The operator $R_0$ is multiplication by $(1+x^2)^{-1}$:
\begin{equation}
    (R_0g)(x)=\frac{g(x)}{1+x^2}.
\end{equation}
Moreover, by the spectral theorem, since $H_\varepsilon\geq0$,
\begin{equation}
    \|R_\varepsilon\|_{L^2\to L^2}
    =
    \sup_{\lambda\in\sigma(H_\varepsilon)}
    \frac{1}{1+\lambda}
    \leq 1.
    \label{eq:resolvent_bound}
\end{equation}
Similarly,
\begin{equation}
    \|R_0\|_{L^2\to L^2}\leq 1.
\end{equation}
These bounds are uniform in $\varepsilon$.

Let $g\in C_c^\infty(\mathbb R)$ and set
\begin{equation}
    u=R_0g=\frac{g}{1+x^2}.
\end{equation}
Then $u\in C_c^\infty(\mathbb R)$, in particular $u''\in L^2(\mathbb R)$,
and
\begin{equation}
    (M+1)u=g.
\end{equation}
Since
\begin{equation}
    (H_\varepsilon+1)u
    =
    (M+1)u-\varepsilon^2u''
    =
    g-\varepsilon^2u'',
\end{equation}
we have
\begin{equation}
    g
    =
    (H_\varepsilon+1)u+\varepsilon^2u''.
\end{equation}
Applying $R_\varepsilon=(H_\varepsilon+1)^{-1}$ gives
\begin{equation}
    R_\varepsilon g
    =
    u+\varepsilon^2R_\varepsilon u''.
\end{equation}
Therefore
\begin{equation}
    R_\varepsilon g-R_0g
    =
    \varepsilon^2R_\varepsilon u''.
\end{equation}
Using the uniform bound \eqref{eq:resolvent_bound}, we obtain
\begin{equation}
    \|R_\varepsilon g-R_0g\|_{L^2}
    \leq
    \varepsilon^2\|u''\|_{L^2}
    \longrightarrow0.
    \label{eq:resolvent_convergence_core}
\end{equation}
Thus $R_\varepsilon g\to R_0g$ for each
$g\in C_c^\infty(\mathbb R)$.

We now extend the convergence to arbitrary $g\in L^2(\mathbb R)$.
Let $g_k\in C_c^\infty(\mathbb R)$ satisfy
\begin{equation}
    \|g-g_k\|_{L^2}\longrightarrow0.
\end{equation}
For fixed $k$,
\begin{equation}
    R_\varepsilon g_k
    \longrightarrow
    R_0g_k
\end{equation}
by Eq.~\eqref{eq:resolvent_convergence_core}. Furthermore, by the uniform
bounds on $R_\varepsilon$ and $R_0$,
\begin{align}
    \|R_\varepsilon g-R_0g\|_{L^2}
    &\leq
    \|R_\varepsilon(g-g_k)\|_{L^2}
    +
    \|R_\varepsilon g_k-R_0g_k\|_{L^2}
    +
    \|R_0(g_k-g)\|_{L^2}
    \nonumber\\
    &\leq
    2\|g-g_k\|_{L^2}
    +
    \|R_\varepsilon g_k-R_0g_k\|_{L^2}.
\end{align}
Taking $\limsup_{\varepsilon\to0}$ gives
\begin{equation}
    \limsup_{\varepsilon\to0}
    \|R_\varepsilon g-R_0g\|_{L^2}
    \leq
    2\|g-g_k\|_{L^2}.
\end{equation}
Finally, sending $k\to\infty$, we obtain
\begin{equation}
    R_\varepsilon g
    \longrightarrow
    R_0g
    \qquad
    \text{for every }g\in L^2(\mathbb R).
\end{equation}
Hence $H_\varepsilon\to M$ in the strong resolvent sense.

We now use the standard convergence theorem for spectral projections:
if $A_j\to A$ in the strong resolvent sense, and $E_j$, $E$ denote the
corresponding spectral measures, then
\begin{equation}
    E_j((-\infty,\lambda])g
    \longrightarrow
    E((-\infty,\lambda])g
\end{equation}
for every $g$, provided $\lambda$ is a continuity point of the spectral
family of $A$, equivalently $E(\{\lambda\})=0$
\cite{ReedSimon1,Kato}.

For the multiplication operator $M=x^2$, the spectral projection onto
$(-\infty,1]$ is multiplication by the characteristic function of the
interval $\{x:x^2\leq1\}$. Therefore
\begin{equation}
    \mathbf 1_{(-\infty,1]}(M)f
    =
    \mathbf 1_{\{x^2\leq1\}}f
    =
    \mathbf 1_{[-1,1]}f.
\end{equation}
Moreover, $1$ is a
continuity point of the spectral family of $M$,a s can be seen from
\begin{equation}
    \mathbf 1_{\{x^2=1\}}f
    =
    \mathbf 1_{\{-1,1\}}f
    =
    0
    \qquad
    \text{in }L^2(\mathbb R),
\end{equation}
because the set $\{-1,1\}$ has Lebesgue measure zero.

Combining the strong resolvent convergence
$H_{\varepsilon_N}\to M$ with the spectral-projection theorem gives
\begin{equation}
    \mathbf 1_{(-\infty,1]}(H_{\varepsilon_N})f
    \longrightarrow
    \mathbf 1_{(-\infty,1]}(M)f
    =
    \mathbf 1_{[-1,1]}f
    \qquad
    \text{in }L^2(\mathbb R).
\end{equation}
Using Eq.~\eqref{eq:f_N_spectral_projection}, this is precisely
\begin{equation}
    f_N
    \longrightarrow
    \mathbf 1_{[-1,1]}f
    \qquad
    \text{in }L^2(\mathbb R).
\end{equation}
This proves the proposition.
\end{proof}

\section{Example of convergence of the rescaled Hermite expansion with a compactly supported function}

Notice that although Appendix~\ref{app:proof} shows that the rescaled Hermite expansion converges to a compactly supported function, it does not necessarily ensure that the smeared propagators also converge when evaluated at the truncated expansion. In this appendix we will show an explicit example that display convergence of the smeared propagators. Explicitly, we will consider an example of the Hermite expansion used to reproduce an entanglement harvesting setup with a compactly supported switching function. In particular, we consider the switching function used in~\cite{patriciaAndI}
\begin{equation}\label{eq:switching-example}
\chi(t) =  \left( 1- \frac{4 t^2}{\bar{T}^2} \right)^{5/2} , \quad t \in [-\bar{T}/2, \bar{T}/2], \chi(t) = 0  \quad \text{otherwise},
\end{equation}
which is compactly supported in $[-\bar{T}/2, \bar{T}/2]$. %In~\cite{patriciaAndI}, the authors consider two inertial comoving finite-sized UDW detectors with spatial shape restricted to spheres of radius $R$, separated by $L = 42R$ with switching functions defined by $T = 40R$, ensuring spacelike separation with $L/T = 1.05$.% where $R$ is the radius of the sphere giving the spatial shape of the detectors with spatial smearing functions supported in a sphere of radius $R$.
For an explicit example, we fix an arbitrary time scale $T_0$ and consider pointlike detectors separated by a distance $L=5 T_0$, as in the main text. For each $N$ we consider the truncated switching functions
\begin{equation}
    \chi_N(t) = \sum_{n=0}^N \langle h_{n,N},\chi\rangle_{L^2}h_{n,N}(t),
\end{equation}
where $h_{n,N}(t) = h_{n}(t,T_N)$, with $T_N$ given by Eqs.~\eqref{eq:TN} and~\eqref{eq:f(N)}. This ensures that in the limit of $N\to \infty$ the detectors are effectively spacelike separated with switching functions supported in $[-L/2,L/2]$.

\begin{figure}[h!]
    \centering
   \includegraphics[width=0.6\linewidth]{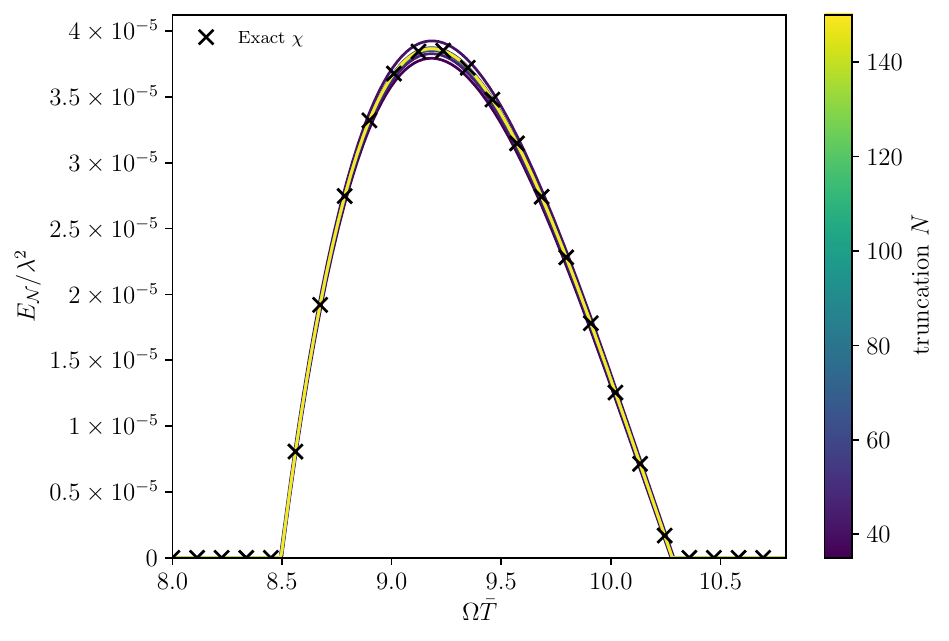}
\vspace{-5mm}
    \caption{Negativity of two pointlike UDW detectors with switching function given by Eq.~\eqref{eq:switching-example}, expanded in the Hermite basis for different truncations, as a function of the energy gap.}
\label{figure:negativity_convergence}
\end{figure}
 In Fig.~\ref{figure:negativity_convergence}, we plot the negativity between the two detectors using the truncated switching functions $\chi_N$ as a function of the energy gap for different truncations of the Hermite expansion and we compare it to the negativity values obtained by considering the non-truncated switching function $\chi$, denoted by black crosses. We see that as $N$ increases, the results using the truncated switching functions converge to the values obtained with the exact form for $\chi$. 
\begin{figure}[h!]
    \centering
    \includegraphics[width=0.49\linewidth]{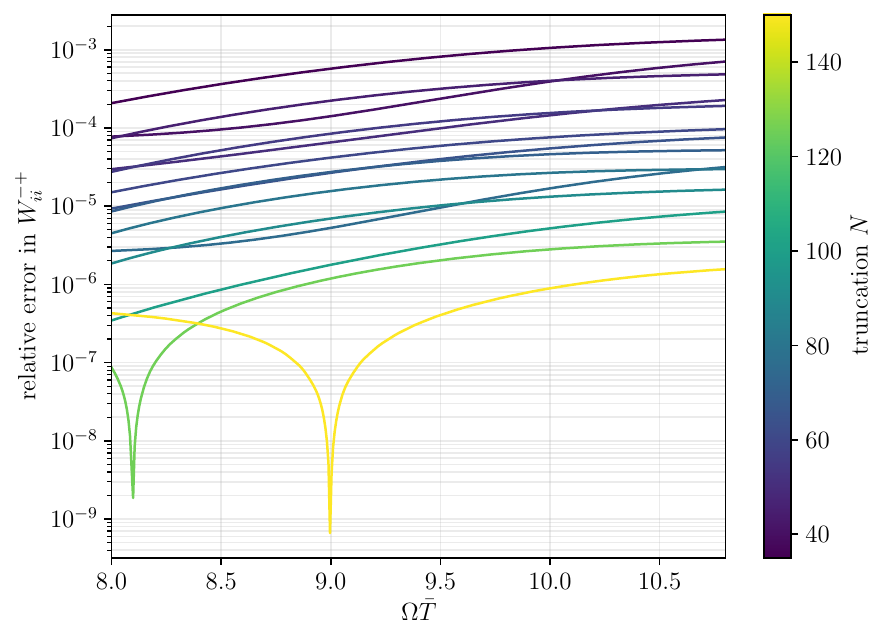}
    \hfill
    \includegraphics[width=0.49\linewidth]{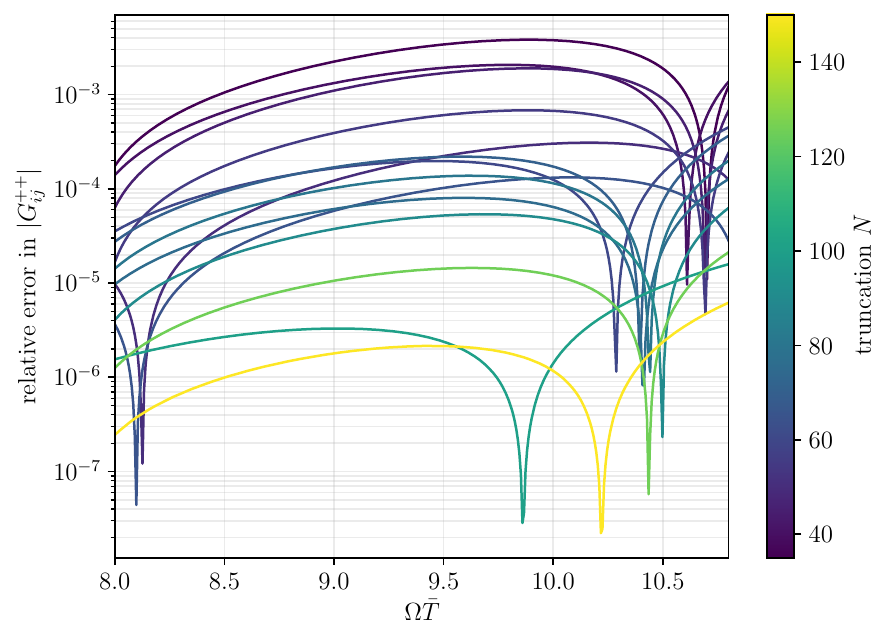}
    \vspace{-5mm}
    \caption{Relative error in the $|G|$ and $W$ terms given by the truncated Hermite expansion.}
    \label{figure:relerror}
\end{figure}

To show convergence, in Fig.~\ref{figure:relerror}, we plot the relative error of the smeared Wightman and Feynman propagators $W(\Lambda_\tc{a}^-,\Lambda_\tc{a}^+) = W(\Lambda_\tc{b}^-,\Lambda_\tc{b}^+)$, $|G_{\tc{f}}(\Lambda_\tc{a}^+,\Lambda_\tc{b}^+)|$ between the exact $\chi$ and its truncated form $\chi_N$ for multiple values of $N$. For $N = 150$ the relative error remains below $10^{-5}$.

%$a \equiv T_0/2$ to approximate their dimensionless ratio $L/(T/2) = 2.1$. The only difference between the two setups is our detectors being pointlike, whose only consequences is that our ratio $L/a$ cannot be exactly $2.1$. In particular, we take $a=2.384 T_0$ and $L/a \approx 2.097$.  \\ \\
%i know you don't like this but i couldn't skip the line otherwise

% In Figs.~\ref{fig:values}, we plot the analytical values of $|\mathcal{M}|$ and $W$ and the values given by the Hermite expansion in the pointlike case.
% We can overall observe that the negativity curves converge to the analytical results and, by $N=150$, the Hermite expansions effectively reproduces the analytical results at very high precision.
% \begin{figure}
%   \centering
%   \includegraphics[width=0.48\linewidth]{Figures/M_values.pdf}
%   %\caption{1a}
%   \label{fig:M_values}
%   \centering
%   \includegraphics[width=0.48\linewidth]{Figures/W_values.pdf}
%   %\caption{1b}
% \caption{Comparison of the $|\mathcal{M}|$ and $W$ terms between the analytical results and the truncated Hermite expansion at $N=150$.}
% \label{fig:values}
% \end{figure}

\end{document}